\documentclass[12pt]{article}
\usepackage{amsmath}
\usepackage{graphicx}
\usepackage{enumerate}
\usepackage{natbib}
\usepackage{url} 

\newcommand{\blind}{0}

\usepackage{amssymb, amsmath, amsthm, amsbsy, bm, bbm, mathrsfs}
	\allowdisplaybreaks[4]
	\newtheorem{Corollary}{Corollary}
	\newtheorem{Lemma}{Lemma}
	
	\newtheorem{Remark}{Remark}
	\newtheorem{Theorem}{Theorem}

\usepackage{algorithm, algpseudocode}
\usepackage{textcomp}
\usepackage{caption, subcaption}
\usepackage{enumitem}
	\setenumerate[1]{itemindent=0pt, itemsep=0pt, partopsep=0pt, parsep=\parskip, topsep=0pt}
\usepackage{color}
\usepackage{array, multirow}

\newcommand{\aalpha}{\boldsymbol{\alpha}}

\newcommand{\dd}{{\rm d}}

\newcommand{\DDD}{\mathbf{D}}
\newcommand{\DDelta}{\boldsymbol{\Delta}}

\newcommand{\II}{\mathbf{I}}
\newcommand{\iid}{{\rm iid}}

\newcommand{\SSigma}{\boldsymbol{\Sigma}}

\DeclareMathOperator*{\argmin}{arg\,min\,}
\DeclareMathOperator*{\cov}{cov}

\DeclareMathOperator*{\E}{E}

\begin{document}

\bibliographystyle{agsm}

\def\spacingset#1{\renewcommand{\baselinestretch}%
{#1}\small\normalsize} \spacingset{1}


\if0\blind
{
  \title{\bf Partial least squares for sparsely observed curves with measurement errors}
  \author{Zhiyang Zhou
	\\
    and \\
    Richard A. Lockhart
	\\
    Department of Statistics \& Actuarial Science, Simon Fraser University
	}
  \maketitle
} \fi

\if1\blind
{
  \bigskip
  \bigskip
  \bigskip
  \begin{center}
    {\LARGE\bf Partial least squares for sparsely observed curves with measurement errors}
\end{center}
  \medskip
} \fi

\bigskip
\begin{abstract}
	Functional partial least squares (FPLS) is commonly used
	for fitting scalar-on-function regression models.
	For the sake of accuracy,
	FPLS demands that 
	each realization of the functional predictor
	is recorded as densely as possible over the entire time span; 
	however, this condition is sometimes violated in,
	e.g., longitudinal studies and missing data research.
	Targeting this point,
	we adapt FPLS to scenarios
	in which the number of measurements per subject 
	is small and bounded from above.
	The resulting proposal is abbreviated as PLEASS.
	Under certain regularity conditions,
	we establish the consistency of estimators 
	and give confidence intervals for scalar responses.
	Simulation studies and real-data applications
	illustrate the competitive accuracy of PLEASS.
\end{abstract}

\noindent%
\vfill

\newpage
\spacingset{1.5} 

\section{Introduction}\label{sec:introduction}

Scalar-on-function (linear) regression (SoFR) is a basic model
in  functional data analysis (FDA).
People have applied it to domains including
chemometrics \citep[e.g.,][]{Goutis1998},
food manufacturing \citep[e.g.,][]{AguileraEscabiasPredaSaporta2010},
geoscience \citep[e.g.,][]{Baillo2009},
medical imaging \citep[e.g.,][]{GoldsmithBobCrainiceanuCaffoReich2011},
and many others.
This model bridges a scalar response $Y$ to a functional predictor $X$ ($=X(\cdot)$),
with the argument of $X$ often referred to as ``time''
and confined to a bounded and closed interval 
$\mathbb{T}\subset\mathbb{R}$.
(Without loss of generality,
we take $\mathbb{T}=[0,1]$ throughout this paper
and omit it in integrals.)
To be specific,
\begin{equation}\label{eq:scalar-on-function}
    Y=\mu_Y+\int\beta(X-\mu_X)+\sigma_{\varepsilon}\varepsilon,
\end{equation}
where: $\mu_X$ (resp. $\mu_Y$) is the expectation of $X$ (resp. $Y$);
the coefficient to be estimated,  $\beta$, belongs to $L^2(\mathbb{T})$ 
(viz. $L^2$-space on $\mathbb{T}$ with respect to (w.r.t.) the Lebesgue measure);
zero-mean noise $\varepsilon$ is of variance one;
and the notation $\int f$ is short for $\int f(t)\dd t$.
The auto-covariance function of $X$ is denoted by
\begin{equation}\label{eq:auto.cov}
    v_A = v_A(s,t)=\cov\{X(s),X(t)\}
\end{equation}
and is assumed to be continuous on $\mathbb{T}^2$.  
Thus $v_A$ has countably many eigenvalues,
say $\lambda_1\geq \lambda_2\geq\cdots$,
such that $\sum_{j=1}^{\infty}\lambda_j = \int v_A(t,t) \, \dd t <\infty$.
Corresponding eigenfunctions are respectively $\phi_1,\phi_2,\ldots$.
In order to ensure the identifiability of $\beta$,
we assume the coefficient function belongs to $\overline{{\rm span}(\phi_1,\phi_2,\ldots)}$,
where ${\rm span}(\cdot)$ denotes the linear space spanned by functions in the parentheses
with the overline representing the closure.
Corresponding to $v_A$,
the auto-covariance operator $\mathcal{V}_A: L^2(\mathbb{T})\to L^2(\mathbb{T})$ 
is defined by,
for each $f\in L^2(\mathbb{T})$,
\begin{equation}\label{eq:V.A}
    \mathcal{V}_A(f)(\cdot) =\int f(t)v_A(t,\cdot)\dd t.
\end{equation}
In this case,
the (Hilbert-Schmidt) operator norm of $\mathcal{V}_A$
equals $\|v_A\|_2$,
viz. the $L^2$-norm of $v_A$.
We abuse $\|\cdot\|_2$ too for the matrix norm induced by the Euclidean norm,
i.e., for arbitrary $\DDD\in\mathbb{R}^{p\times q}$ and $\aalpha\in\mathbb{R}^{q\times 1}$,
$\|\DDD\|_2=\sup_{\aalpha:\|\aalpha\|_2=1}\|\DDD\aalpha\|_2$.
It is well known that 
$\|\DDD\|_2$ is actually the largest eigenvalue of $\DDD$
and reduces to the Euclidean norm for vectors.

The typical first step in estimating $\beta$
is to project it onto a space spanned by basis functions
either fixed (e.g., wavelets or splines) or data-driven
(e.g., functional principal component (FPC) or functional partial least squares (FPLS)).
There are already numerous studies comparing FPC and FPLS
(e.g., \citealp{ReissOgden2007, AguileraEscabiasPredaSaporta2010}).
They concluded that
FPLS is superior to FPC in the sense that
the former  
provides a more accurate parameter estimation
and yields more parsimonious models
\citep[][pp.~53]{Albaqshi2017}.

\subsection{Introduction to functional partial least squares}
 
Partial least squares (PLS) is a name shared by diverse algorithms in the multivariate context,
including nonlinear iterative PLS (NIPALS, \citealp{Wold1975}) and the
statistically inspired modification of PLS (SIMPLS, \citealp{deJong1993})
as two of the most well-known.
Analogously,
the implementation of FPLS is far from unique:
it constructs basis functions by recursively maximizing
(the functional version of) Tucker's criterion \citep[see, e.g., Proposition~1 of][for its expression]{PredaSaporta2005}
subject to various orthonormality constraints.
For FoFR,
\cite{DelaigleHall2012b} observed the equivalence between functional extensions of NIPALS and SIMPLS:
the first $p$ basis functions arising via these two distinct routes span spaces identical to
the functional version of $p$-dimensional Krylov subspace (KS),
namely,
\begin{equation}\label{eq:KSp}
    {\rm KS}_p
    ={\rm span}\{\mathcal{V}_A(\beta),\ldots,\mathcal{V}_A^p(\beta)\}
    ={\rm span}\{v_C,\ldots,\mathcal{V}_A^{p-1}(v_C)\},
\end{equation}
where $\mathcal{V}_A^j$ is the $j$th power of $\mathcal{V}_A$, and
\begin{equation}\label{eq:cross.cov}
    v_C=v_C(\cdot)=\cov\{Y,X(\cdot)\}.
\end{equation}
To be explicit,
starting with $\mathcal{V}_A^1=\mathcal{V}_A$,
we define recursively, 
$\mathcal{V}_A^j: L^2(\mathbb{T})\to L^2(\mathbb{T})$ by
\begin{equation}\label{eq:V.A.j}
    \mathcal{V}_A^j(f)(\cdot)
    =\int \mathcal{V}_A^{j-1}(f)(t)v_A(t,\cdot)\dd t,
    \quad\forall f\in L^2(\mathbb{T}).
\end{equation}

\citet[Theorem~3.2]{DelaigleHall2012b} 
showed that
$\beta$ must be located in ${\rm KS}_{\infty}=\overline{{\rm span}\{v_C,\mathcal{V}_A(v_C),\ldots\}}$. Hence $\beta$ is the limit (in the $L^2$ sense) of 
$$
	\beta_p
	=\argmin_{\theta\in {\rm KS}_p}
  		\E\left\{Y-\mu_Y-\int\theta(X-\mu_X)\right\}^2.
$$
Once we obtain $w_1,\ldots,w_p$ by 
(modified-Gram-Schmidt) 
orthonormalizing $\mathcal{V}_A(\beta),\ldots,\mathcal{V}_A^p(\beta)$ w.r.t. $v_A$
(following Algorithm \ref{alg:mgs} below or \citealp[pp.~102]{Lange2010}),
$\beta_p$ can then be rewritten as
\begin{equation}\label{eq:beta.p}
	\beta_p
 	= [w_1,\ldots,w_p]\bm{c}_p,
\end{equation}
where
\begin{equation}
	\bm{c}_p 
	= \left[
	    \int w_1\mathcal{V}_A(\beta), \ldots,
	    \int w_p\mathcal{V}_A(\beta)
	\right]^\top
	= \left[
	    \int w_1 v_C, \ldots,
	    \int w_p v_C
	\right]^\top.
	\label{eq:c.p}
\end{equation}

Now consider a new pair $(X^*,Y^*) \sim (X,Y)$. Then 
as $p \to \infty$,
\begin{equation}\label{eq:eta.p}
	\eta_p(X^*)=\mu_Y+\int \beta_p(X^*-\mu_X)
	=\mu_Y+[\xi_1^*,\ldots,\xi_p^*]\bm{c}_p
\end{equation}
approaches the conditional expectation of $Y^*$ given $X^*$,
viz.
\begin{equation}\label{eq:eta.x.star}
	\eta(X^*)=\E(Y^*\mid X^*)=\mu_Y+\int \beta(X^*-\mu_X)
\end{equation}
in which
\begin{equation}\label{eq:xi.j.star}
	\xi_j^*=\int w_j(X^*-\mu_X).
\end{equation}
We  refer  to $\xi_j^*$ as the $j$th FPLS score (associated with $X^*$).
(Henceforth superscript * indicates items associated with the new realization $X^*\sim X$.)
Plugging empirical counterparts into \eqref{eq:beta.p} and \eqref{eq:eta.x.star},
the proposal of \citet[][Section~4]{DelaigleHall2012b}
is equivalent (in terms of estimating $\beta$ as well as predicting $Y^*$)
to functional counterparts of NIPALS and SIMPLS.

\begin{algorithm}[!t]
	\caption{Orthonormalize $\psi_1,\ldots,\psi_p\in L^2(\mathbb{T})$
	    into $\vartheta_1,\ldots,\vartheta_p\in L^2(\mathbb{T})$
	    w.r.t. $\varrho\in L^2(\mathbb{T}^2)$
	}
	\label{alg:mgs}
	\begin{algorithmic}[]
		\For {$j$ in $1,\ldots,p$}
		    \State $\vartheta_j^{[1]}\gets\psi_j$.
		    \If {$j\geq 2$}
    		    \For {$i$ in $1,\ldots,j-1$}
    		        \State $\vartheta_j^{[i+1]}\gets
    		            \vartheta_j^{[i]}-
    		                \vartheta_i\int\int\vartheta_j^{[i]}(s)\varrho(s,t)\vartheta_i(t)$.
    		    \EndFor
    		\EndIf
    		\If {$\int\int\vartheta_j^{[j]}(s)\varrho(s,t)\vartheta_j^{[j]}(t)>
    		    \text{preset small positive threshold}$}
    		    \State $\vartheta_j\gets\vartheta_j^{[j]}/
    		        \{\int\int\vartheta_j^{[j]}(s)\varrho(s,t)\vartheta_j^{[j]}(t)\}^{1/2}$.
    		\Else \State $\vartheta_j\gets0$.
    		\EndIf
		\EndFor
	\end{algorithmic}
\end{algorithm}

\subsection{Sparsity and measurement errors}
Like most FDA techniques,
FPLS algorithms are designed for dense settings,
i.e.,
realizations of $X$ are supposed to be densely observed,
since their implementations inevitably involve approximations to integrals.
This condition is not expected to be fulfilled under all circumstances.
For example, in typical clinical trials,
participants cannot be monitored 24/7;
instead,
they are required to visit the clinic repeatedly on specific dates.
Due to cost and convenience,
the scheduled visiting frequency is doomed to be sparse for essentially every subject.
What is  worse is that
subjects tend to show up on their own basis
with frequencies lower and more irregular than scheduled.
Similar difficulties can arise in missing data problems
where a number of recordings are lost for whatever reason.

The training sample consists of 
$n$ two-tuples
$(X_1, Y_1),\ldots,(X_n, Y_n)$ independently and identically distributed (iid) as $(X, Y)$.
Specifying the sparsity and measurement errors simultaneously,
we suppose the $i$th trajectory is measured at only $L_i$ (random) time points
(say $T_{i1},\ldots,T_{iL_i}$)
with corresponding contaminated observations
\begin{equation}\label{eq:observation}
    \widetilde{X}_i(T_{i\ell})=X_i(T_{i\ell})+\sigma_e e_{i\ell},
    \quad\ell = 1,\ldots, L_i,
\end{equation}
where $\sigma_e>0$ and the $e_{i\ell}$ are white noise with mean zero and variance one.
We assume that all the time points and error terms 
are independent across subjects and from each other.
More rigorous description is detailed in Appendix \ref{appendix:technical}.
This joint setup of sparsity and error-in-variable is also considered in existing literature
including but not limited to \cite{YaoMullerWang2005a, YaoMullerWang2005b},
\cite{XiaoLiCheckleyCrainiceanu2018},
and \cite{RubinPanaretos2020}.

\begin{Remark}
    For each $i$, 
    it is not necessary to order $T_{i1},\ldots,T_{iL_i}$ in a specific way.
    Additionally we suggest not viewing $\widetilde{X}_i$ as the sum of $X_i$ and a white noise process,
    otherwise more mathematical effort is needed in the definition to ensure rigor.
    We utilize only (univariate) random variables 
    $\widetilde{X}_i(T_{i1}),\ldots,\widetilde{X}_i(T_{iL_i})$
    and never attempt to approximate integrals involving an entire function $\widetilde{X}_i$.
\end{Remark}

As pioneers who extended classical FPC to this challenging setting,
\citet{JamesHastieSugar2000} postulated a reduced rank mixed effects model 
fitted by the expectation-maximization algorithm and penalized least squares.
Abbreviated as PACE,
the proposal of \citet{YaoMullerWang2005a, YaoMullerWang2005b}
introduces a local linear smoother (LLS) 
estimator for $v_A$
followed by FPC scores $\rho_j$ ($=\int\phi_j(X_i-\mu_X)$)
which are approximated by conditional expectations.
To the best of our knowledge,
there are still few extensions of FPLS 
applicable to such a scenario.
In this work,
we attempt to fill in this blank 
by developing a new technique named  Partial LEAst Squares for Sparsity (PLEASS),
handling sparse observations and measurement errors simultaneously.

Here is a sketch of the procedure for PLEASS.
First,
thanks to the iid assumption on subjects,
we are able to pool together all the observations 
in order to recover the variance and covariance functions
from which basis functions are extracted.
Then, $\beta$ is estimated by 
plugging empirical counterparts into $\beta_p$ at \eqref{eq:beta.p}.
It is worth noting that,
since $X^*$ is not observed densely, 
PLEASS does not give a consistent prediction for $\eta(X^*)$ at \eqref{eq:eta.x.star};
instead it constructs a confidence interval (CI) for $\eta(X^*)$
through conditional expectation. 

The remainder of this paper is organized as follows.
Section \ref{sec:method} details the implementation procedure for PLEASS.
In Section \ref{sec:theory},
we present asymptotic results on the consistency of estimators and
on the distribution of $\eta(X^*)$.
Section \ref{sec:numerical} applies PACE and PLEASS to both simulated and authentic datasets
and compares their resulting performances. 
Concluding remarks are given in Section \ref{sec:conclusion}.
Finally  we include more technical arguments in appendices.

\section{Methodology}\label{sec:method}

\subsection{Estimation and prediction}\label{sec:estimation}

The first phase of PLEASS is to find estimators for
$\mu_X$ at \eqref{eq:scalar-on-function},
$v_A$ at \eqref{eq:auto.cov},
$v_C$ at \eqref{eq:cross.cov},
and $\sigma_e^2$ at \eqref{eq:observation},
respectively,
say, 
$\hat{\mu}_X$, $\hat{v}_A$, $\hat{v}_C$ and $\hat{\sigma}_e^2$.
Existing methods for 
reconstructing the variance and covariance structure from sparse observations
roughly fall into three categories:
i) kernel smoothing 
(e.g., \citealp[LLS in][]{YaoMullerWang2005a, YaoMullerWang2005b} and \citealp{LiHsing2010}
and \citealp[the modified kernel smoothing in][]{PaulPeng2011}),
ii) spline smoothing
(e.g., fast covariance estimation (FACE) by \citealp{XiaoLiCheckleyCrainiceanu2018}),
and iii) maximum likelihood
(ML, e.g., \citealp[restricted ML in][]{JamesHastieSugar2000} and \citealp{PengPaul2009}
and \citealp[quasi-ML in][]{ZhouLinLiang2018}).
Typically,
the third category requires initial values obtained through the first two
and is hence more time-consuming.
In the numerical study (Section \ref{sec:numerical} below),
we adopt both
LLS 
(whose details are relegated to Appendix \ref{appendix:LLS},
following \citealp{YaoMullerWang2005a, YaoMullerWang2005b}) and FACE.
LLS,
which is also exploited by PACE,
has nice asymptotic properties \citep{HallMullerWang2006},
whereas FACE runs faster and has competitive accuracy.

\begin{Remark}
    In theory,
    the framework of PLEASS is flexible as to how to estimate
    $\mu_X$,
    $v_A$,
    $v_C$,
    and $\sigma_e^2$,
    as long as
    $\|\hat{\mu}_X-\mu_X\|_{\infty}$,
    $\|\hat{v}_A-v_A\|_{\infty}$,
    $\|\hat{v}_C-v_C\|_{\infty}$,
    and $|\hat{\sigma}_e^2-\sigma_e^2|$
    all converge to zero as $n$ diverges
    (with $\|\cdot\|_{\infty}$ denoting the $L^{\infty}$-norm).
    It is even more flexible in practice
    and permits any way of recovery preferred by users.
    Theoretical results in upcoming Section \ref{sec:theory} 
    are merely demos corresponding to LLS; our results can be adapted to other approaches.
\end{Remark}

It is understood that in numerical implementation integrals have to be approximated by,
e.g., quadrature rules.
\citet{Tasaki2009} gave upper bounds on the (absolute) approximation errors for
Riemann and trapezoidal sums; 
these bounds
tend to zero as the discretized grid becomes dense.
We hereafter use (abuse) the integral notation even for corresponding numerical approximations.

Recursively define the empirical counterpart of $\mathcal{V}_A^j$ at \eqref{eq:V.A.j} by
\begin{equation}\label{eq:V.hat.A}
    \widehat{\mathcal{V}}_A^j(f)(\cdot)
    =\int \widehat{\mathcal{V}}_A^{j-1}(f)(t)\hat{v}_A(t,\cdot)\dd t,
    \quad\forall f\in L^2(\mathbb{T}).
\end{equation}
The recursion is initialized by taking $\mathcal{V}^0$ to be the identity operator.
Then orthogonal basis functions $\hat{w}_1,\ldots,\hat{w}_p$
are constructed from 
$\hat{v}_C,\widehat{\mathcal{V}}_A(\hat{v}_C),\ldots,\widehat{\mathcal{V}}_A^{p-1}(\hat{v}_C)$
(following Algorithm \ref{alg:mgs} or \citealp[pp.~102]{Lange2010}).
Evidently a plug-in estimator for $\beta$ is given by
\begin{equation}\label{eq:beta.p.hat}
	\hat{\beta}_p = [\hat{w}_1,\ldots,\hat{w}_p]\hat{\bm{c}}_p.
\end{equation}
These estimators converge to the true $\beta$ as $n$ and $p$, respectively, diverge at specific rates
(see Theorem \ref{thm:converge.beta}),
with
\begin{equation}\label{eq:c.p.hat}
	\hat{\bm{c}}_p 
	= \left[
	    \int \hat w_1 \hat v_C, \ldots,
	    \int \hat w_p \hat v_C
	\right]^\top.
\end{equation}
estimating \eqref{eq:c.p}.

Predicting $\eta(X^*)$ at \eqref{eq:eta.x.star} is a problem fairly different from estimation.
Since $\widetilde{X}^*$ (viz. the contaminated $X^*$) is only observed 
at $L^*$ ($\sim L$) time points,
it is not practical to numerically integrate 
the product of $\hat{\beta}_p$ and $\widetilde{X}^*$.
Instead we target the prediction of a surrogate for $\eta(X^*)$. 
That surrogate, denoted by $\tilde{\eta}_{\infty}(X^*)$, is  defined at \eqref{eq:eta.inf.tilde} below. 
Write
\begin{align}
	\widetilde{\bm{X}}^* &=\left[\widetilde{X}^*(T_1^*),\ldots,\widetilde{X}^*(T_{L^*}^*)\right]^\top, 
	\notag\\
	\bm{\mu}_X^* &= \E(\widetilde{\bm{X}}^*|L^*, T_1^*,\ldots,T_{L^*}^*)
	    =\left[\mu_X(T_1^*),\ldots,\mu_X(T_{L^*}^*)\right]^\top,
	\notag\\
	\SSigma_{\widetilde{X}^*} 
	&=\left[v_A(T_{\ell_1}^*,T_{\ell_2}^*)\right]_{1\leq\ell_1,\ell_2\leq L^*}+\sigma_e^2\II_{L^*},
	\notag
\end{align}
and, for integer $j\in[1,p]$,
$$
	\bm{h}_j^* =\left[\mathcal{V}_A(w_j)(T_1^*),\ldots,\mathcal{V}_A(w_j)(T_{L^*}^*)\right]^\top.
$$
Conditional on $L^*$ and $T_1^*,\ldots,T_{L^*}^*$,
in view of the identity
$$
    \cov(\widetilde{\bm{X}}^{*\top}, \xi_1^*, \ldots, \xi_p^*|L^*, T_1^*,\ldots,T_{L^*}^*) =
        \left[\begin{array}{cccc} 
            \SSigma_{\widetilde{X}^*}	&\bm{h}_1^*	&\cdots	&\bm{h}_p^*
            \\ 
            \bm{h}_1^{*\top}
            \\
            \vdots	&	&\II_p
            \\
            \bm{h}_p^{*\top}
        \end{array}\right],
$$
the best linear unbiased prediction for $\xi_j^*$ is
\begin{equation}\label{eq:xi.j.tilde.star}
    \tilde{\xi}_j^*
    =\E(\xi_j^*\mid\widetilde{\bm{X}}^{*\top}, L^*, T_1^*,\ldots,T_{L^*}^*)
    =\bm{h}_j^{*\top}\SSigma_{\widetilde{X}^*}^{-1}(\widetilde{\bm{X}}^*-\bm{\mu}_X^*).
\end{equation}
This predictor minimizes 
$
    \E[\{\xi_j^*-f(\widetilde{\bm{X}}^*)\}^2 | L^*, T_1^*,\ldots,T_{L^*}^*] 
$
over all linear functions $f$ subject to 
$\E\{\xi_j^*-f(\widetilde{\bm{X}}^*)\mid L^*, T_1^*,\ldots,T_{L^*}^*\}=0$.
It is even the best prediction over all  measurable $f$, linear or not, 
as long as $\xi_1^*,\ldots,\xi_p^*$ and $\widetilde{\bm{X}}^*$ are jointly Gaussian
\citep[][Theorem 1]{Harville1976}.
Geometrically speaking, 
$\tilde{\xi}_j^*$ at \eqref{eq:xi.j.tilde.star}
is the (orthogonal) projection of $\xi_j^*$ at \eqref{eq:xi.j.star} onto
${\rm span}\{\widetilde{X}^*(T_1^*),\ldots,\widetilde{X}^*(T_{L^*}^*)\}$
(given $L^*$ and $T_1^*,\ldots,T_{L^*}^*$).
Then the projection of $\eta_p(X^*)-\mu_Y$ onto the same space is
$\tilde{\eta}_p(X^*)-\mu_Y$;
recall $\eta_p(X^*)$ is defined at \eqref{eq:eta.p}.  If we define the 
$L^*\times p$ matrix $\mathbf{H}_p^*=[\bm{h}_1^*,\ldots,\bm{h}_p^*]$,
then we have 
\begin{equation}\label{eq:eta.p.tilde}
    \tilde{\eta}_p(X^*)
        =\mu_Y + [\tilde{\xi}_1^*,\ldots,\tilde{\xi}_p^*]\bm{c}_p
        =\mu_Y
        	+\bm{c}_p^\top
        		\mathbf{H}_p^{*\top}\SSigma_{\widetilde{X}^*}^{-1}(\widetilde{\bm{X}}^*-\bm{\mu}_X^*).
\end{equation}
Accordingly,
\begin{equation}\label{eq:eta.inf.tilde}
	\tilde{\eta}_{\infty}(X^*) = \lim_{p\to\infty}\tilde{\eta}_p(X^*)
\end{equation}
is a natural surrogate for $\eta(X^*)$ at \eqref{eq:eta.x.star}.

It is therefore justified to predict $Y^*$
by the empirical counterpart of \eqref{eq:eta.p.tilde}, namely,
\begin{equation}\label{eq:eta.p.hat}
	\hat{\eta}_p(X^*)
    =\bar{Y}
       	+\hat{\bm{c}}_p^\top
       		\widehat{\mathbf{H}}_p^{*\top}
       		\widehat{\SSigma}_{\widetilde{X}^*}^{-1}(\widetilde{\bm{X}}^*-\hat{\bm{\mu}}_X^*),
\end{equation}
which is constructed by replacing population quantities 
$\mu_Y$, $\bm{c}_p$, $\SSigma_{\widetilde{X}^*}$, $\bm{\mu}_X^*$, and $\mathbf{H}_p^*$
all at \eqref{eq:eta.p.tilde} with,
respectively,
$\bar{Y} = n^{-1}\sum_{i=1}^n Y_i$, $\hat{\bm{c}}_p$ at \eqref{eq:c.p.hat}, and
\begin{align}
    \widehat{\SSigma}_{\widetilde{X}^*} 
        &=\left[\hat{v}_A(T_{\ell_1}^*,T_{\ell_2}^*)\right]_{1\leq\ell_1,\ell_2\leq L^*}
            +\hat{\sigma}_e^2\II_{L^*},
    \label{eq:sigma.x.hat}\\
    \hat{\bm{\mu}}_X^* 
        &=\left[\hat{\mu}_X(T_1^*),\ldots,\hat{\mu}_X(T_{L^*}^*)\right]^\top,
    \label{eq:mu.x.hat.star}\\
   	\widehat{\mathbf{H}}_p^*
   		&=[\widehat{\mathcal{V}}_A(\hat{w}_j)(T_{\ell}^*)]_{
   		    \substack{1\leq j\leq p \\ 1\leq\ell\leq L^*}
   		  }.
   	\label{eq:H.p.hat}
\end{align}

It remains to construct a CI for $\eta(X^*)$ at \eqref{eq:eta.x.star}.
From the perspective of projection again,
we have
\begin{align*}
    \cov([\xi_1^*-\tilde{\xi}_1^*,&\ldots,\xi_p^*-\tilde{\xi}_p^*]^\top
        \mid L^*, T_1^*,\ldots,T_{L^*}^*)
    \\
    =&\ \cov([\xi_1^*,\ldots,\xi_p^*]^\top \mid L^*, T_1^*,\ldots,T_{L^*}^*)
    -\cov([\tilde{\xi}_1^*,\ldots,\tilde{\xi}_p^*]^\top \mid L^*, T_1^*,\ldots,T_{L^*}^*)
    \\
    =&\ \II_p-\mathbf{H}_p^{*\top}\SSigma_{\widetilde{X}^*}^{-1}\mathbf{H}_p^*.
\end{align*}
Under Gaussian assumptions (as in Corollary \ref{thm:asymptotic.normal})
and conditioning on $L^*$ and $T_1^*,\ldots,T_{L^*}^*$, 
the error
$\hat{\eta}_p(X^*)-\eta(X^*)$ is asymptotically normally distributed.
An asymptotic $(1-\alpha)$ (conditional Wald) CI for $\eta(X^*)$ at \eqref{eq:eta.x.star} is then
$$
    \hat{\eta}_p(X^*)\pm \Phi^{-1}_{1-\alpha/2}\left\{
        \hat{\bm{c}}_p^\top
        (
            \II_p
            -\widehat{\mathbf{H}}_p^{*\top}\widehat{\SSigma}_{\widetilde{X}^*}^{-1}\widehat{\mathbf{H}}_p^*
        )
        \hat{\bm{c}}_p
    \right\}^{1/2},    
$$
where $\Phi^{-1}_{1-\alpha/2}$ is the $(1-\alpha/2)$ standard normal quantile.

\subsection{Selection of number of basis functions}

We are unclear on how to estimate the degrees of freedom (DoF) asscociated with 
PLEASS prediction $\hat{\eta}_p(X^*)$ at \eqref{eq:eta.p.hat},
partially because of its intrinsic complexity;
at least there seems no natural extension from 
the work of \cite{KramerSugiyama2011}
on DoF computation for (multivariate) PLS.
As a consequence, 
rather than using generalized cross validation \citep{CravenWahba1979}
and various information criteria,
it sounds more reasonable to employ  (leave-one-out) cross-validation (CV) as the tuning scheme:
choose an integer $p\in [0,p_{\max}]$ by minimizing
$$
    {\rm CV}(p)=n^{-1}\sum_{i=1}^n\{Y_i-\hat{\eta}_p^{(-i)}(X_i)\}^2
$$
in which $\hat{\eta}_p^{(-i)}(X_i)$ predicts the $i$th response
with all the other subjects kept for training.
Define by ${\rm FVE}(j)=\sum_{k=1}^j\lambda_k / \sum_{k=1}^{\infty}\lambda_k$
(with $\lambda_k$ replaced by empirical counterparts in practice)
the fraction of variance explained (FVE) by the first $j$ eigenfunctions.
An upper bound for $p$ is then given by,
e.g.,
\begin{equation}\label{eq:pfve}
    p_{\max}=\min\{j\in\mathbb{Z}^+: {\rm FVE}(j)\geq 95\%\}.
\end{equation}
This cut-off is one of the default truncation rules frequently used for the Karhunen-Lo\`{e}ve series.
Since,
as mentioned in Section~\ref{sec:introduction},
FPLS typically needs fewer terms than FPC 
to reach a comparable accuracy,
\eqref{eq:pfve} is very likely to be large enough for tuning PLEASS.
Another heuristic upper bound is provided by \citet[][Section 3]{DelaigleHall2012a}:
$p_{\max}=n/2$,
acceptable for a small or moderate $n$.

\begin{algorithm}[t!]
	\caption{PLEASS tuned through CV}
	\label{alg:pleass}
	\begin{algorithmic}[]
	    \State Obtain $\hat{\mu}_X$, $\hat{v}_A$, $\hat{v}_C$ and $\hat{\sigma}_e^2$
	        following Appendix \ref{appendix:LLS}.
		\For {$j$ in $1,\ldots,p_{\max}-1$}
            \State $\widehat{\mathcal{V}}_A^j(\hat{v}_C)(\cdot)
                    \gets\int_{\mathbb{T}}
                        \hat{v}_A(\cdot, t)\widehat{\mathcal{V}}_A^{j-1}(\hat{v}_C)(t)\dd t.
                $
		\EndFor
	    \State Extract $\hat{w}_j$ from 
	        $\hat{v}_C,
	        \widehat{\mathcal{V}}_A(\hat{v}_C)
	        ,\ldots,
	        \widehat{\mathcal{V}}_A^{p_{\max}-1}(\hat{v}_C)$
	        following Algorithm \ref{alg:mgs}.
	    \State $\hat{\beta}_0\gets0$.
        \State $\hat{\eta}_0(X^*)\gets \bar{Y}$.
		\For {$p$ in $1,\ldots,p_{\max}$}
            \State 
            	$\hat{\beta}_p\gets
            		[\hat{w}_1,\ldots,\hat{w}_p]\hat{\bm{c}}_p$
            	with $\hat{\bm{c}}_p$ at \eqref{eq:c.p.hat}.
		    \State $\hat{\eta}_p(X^*)\gets
    		            \bar{Y}
	                   	+\hat{\bm{c}}_p^\top
	                   		\widehat{\mathbf{H}}_p^{*\top}
	                   		\widehat{\SSigma}_{\widetilde{X}^*}^{-1}(\widetilde{\bm{X}}^*-\hat{\bm{\mu}}_X^*)
                    $
            \State\qquad with $\widehat{\SSigma}_{\widetilde{X}^*}$ at \eqref{eq:sigma.x.hat},
                    $\hat{\bm{\mu}}_X^*$ at \eqref{eq:mu.x.hat.star}
                    and $\widehat{\mathbf{H}}_p^*$ at \eqref{eq:H.p.hat}.
		\EndFor
		\State $p_{\rm opt}\gets\argmin_{0\leq p\leq p_{\max}}{\rm CV}(p)$.
		\State 
		$(1-\alpha)$ CI for $\eta(X^*)
		\gets\hat{\eta}_{p_{\rm opt}}(X^*)\pm \Phi^{-1}_{1-\alpha/2}\left\{
                \hat{\bm{c}}_{p_{\rm opt}}^\top
                (\II_{p_{\rm opt}}
                    -\widehat{\mathbf{H}}_{p_{\rm opt}}^\top
                    \widehat{\SSigma}_{\widetilde{X}^*}^{-1}
                    \widehat{\mathbf{H}}_{p_{\rm opt}})
                \hat{\bm{c}}_{p_{\rm opt}}
            \right\}^{1/2}
		$.
	\end{algorithmic}
\end{algorithm}

\section{Asymptotic properties}\label{sec:theory}

Our theoretical results are established under \ref{cond:rv.1}--\ref{cond:tau.p.sup}
in Appendix \ref{appendix:technical}. 
The first six of these assumptions formalize the setup of sparsity and measurement errors;
\ref{cond:auto.cov}--\ref{cond:bandwidth.4}  
are prepared for the consistency of LLS in Appendix \ref{appendix:LLS}.
For arbitrary fixed $p$,
the consistency of $\hat{\beta}_p$ at \eqref{eq:beta.p.hat} 
is a direct corollary of \citet[][Theorem~1]{Zhou2019}.
Unfortunately,
this argument may not apply to the scenario with diverging $p=p(n)$,
since the sequential construction in \eqref{eq:V.hat.A}
tends to induce a bias accumulating with increasing $p$.
As a result,
it is indispensable to impose a sufficiently slow divergence rate on $p$,
such as, e.g., those required by \ref{cond:tau.p.L2} or \ref{cond:tau.p.sup}.

\begin{Theorem}\label{thm:converge.beta}
    Assume that \ref{cond:rv.1}--\ref{cond:tau.p.L2} all hold.
    As $n$ goes to infinity,
    $\|\hat{\beta}_p-\beta\|_2\to_p 0$.
    If we substitute the stronger assumption \ref{cond:tau.p.sup} for \ref{cond:tau.p.L2},
    and assume, in addition, that $\|\beta_p-\beta\|_{\infty}\to_p 0$,
    then the convergence of $\hat{\beta}_p$ becomes uniform, 
    i.e.,
    $\|\hat{\beta}_p-\beta\|_{\infty}\to_p 0$.
\end{Theorem}

Analogous to PACE,
our PLEASS results in an inconsistent prediction (see Theorem \ref{thm:converge.eta}):
the discrepancy  
$\hat{\eta}_p(X^*)-\tilde{\eta}_{\infty}(X^*)$ between our forecast and our surrogate converges to zero (unconditionally and in probability) 
but not the discrepancy  $\hat{\eta}_p(X^*)-\eta(X^*)$ between our forecast and the true mean of $Y^*$.
Nevertheless,
this phenomenon is far from disappointing:
one implication is that 
$\hat{\eta}_p(X^*)-\eta(X^*)$ is asymptotically distributed as
$\tilde{\eta}_{\infty}(X^*)-\eta(X^*)$;
an asymptotic distribution of $\eta(X^*)$ hence follows.
In particular,
the result for Gaussian cases is presented in Corollary \ref{thm:asymptotic.normal}.

\begin{Theorem}\label{thm:converge.eta}
    Under \ref{cond:rv.1}--\ref{cond:tau.p.L2},
    as $n$ goes to infinity,
    $\hat{\eta}_p(X^*)-\tilde{\eta}_{\infty}(X^*)$ converges to zero (unconditionally) in probability.
\end{Theorem}

\begin{Corollary}\label{thm:asymptotic.normal}
    Fix $L^*$ and $T_1^*,\ldots,T_{L^*}^*$.
    Assume \ref{cond:rv.1}--\ref{cond:tau.p.L2}
    as well as the following two extra conditions:
    \begin{enumerate}[label=\arabic*)]
        \item
            FPLS scores $\int w_j(X-\mu_X)$
            and measurement errors $e_{i\ell}$
            are jointly Gaussian.
        \item
            $\lim_{p\to\infty}\bm{c}_p^\top
                (\II_p
                    -\mathbf{H}_p^{*\top}\SSigma_{\widetilde{X}^*}^{-1}\mathbf{H}_p^*
                )\bm{c}_p
                =\omega>0
            $.
    \end{enumerate}
    Then,
    as $n\to\infty$,
    $$
	    \frac{\hat{\eta}_p(X^*)-\eta(X^*)}{
	        \sqrt{
	            \hat{\bm{c}}_p^\top
	            (\II_p
	                -\widehat{\mathbf{H}}_p^{*\top}
	                \widehat{\SSigma}_{\widetilde{X}^*}^{-1}
	                \widehat{\mathbf{H}}_p^*)
	            \hat{\bm{c}}_p
	        }
	    }\to_d\mathcal{N}(0,1).
    $$
\end{Corollary}

\section{Numerical illustration}\label{sec:numerical}

PLEASS is compared here with PACE in terms of finite-sample numerical performance.
As mentioned in Section \ref{sec:estimation}, 
both LLS and FACE
(implemented respectively via 
\texttt{R} packages
\texttt{fdapace} \citep{R-fdapace}
and \texttt{face} \citep{R-face})
were utilized to eestimate population quantities
$\mu_X$ at \eqref{sec:introduction},
$v_A$ at \eqref{eq:auto.cov},
$v_C$ at \eqref{eq:cross.cov},
and $\sigma_e^2$ at \eqref{eq:observation}.
Resulting combinations, 
viz. PLEASS+LLS, PACE+LLS, PLEASS+FACE and PACE+FACE,
are abbreviated as PLEASS.L, PACE.L, PLEASS.F and PACE.F,
respectively.
Our code trunks are accessible at 
\url{https://github.com/ZhiyangGeeZhou/PLEASS}.

\subsection{Simulation}\label{sec:simulation}

Each sample consisted of
$n=300$ iid paired realizations of $(X, Y)$
with $X$ and $Y$ both of zero mean.
$X$ was set up as a Gaussian process,
i.e., $\lambda_j^{-1/2}\rho_j=\lambda_j^{-1/2}\int\phi_j(X-\mu_X)$ were all iid as standard normal.
Error terms $e_{i\ell}$ were also standard normal.
We took 100, 90, 80, 10, 9, 8, 1, 0.9, and 0.8
as the top nine eigenvalues of operator $\mathcal{V}_A$ at \eqref{eq:V.A}; all the rest
were 0.
Correspondingly,
the top nine eigenfunctions were taken to be
(normalized) shifted Legendre polynomials 
\citep[refer to][pp.~773--774]{Hochstrasser1972}
of order 1 to 9,
say $P_1,\ldots,P_9$;
unit-normed and mutually orthogonal on $[0,1]$,
they were generated through \texttt{R}-package \texttt{orthopolynom} \citep{R-orthopolynom}.
The slope function $\beta$ was given by one of the following cases:
\begin{align}
    \beta &= P_1 + P_2 + P_3, \label{eq:beta.simu.1}
    \\
    \beta &= P_4 + P_5 + P_6,\label{eq:beta.simu.2}
    \\
    \beta &= P_7 + P_8 + P_9.\label{eq:beta.simu.3}
\end{align}
Two sorts of signal-to-noise-ratio (SNR) were defined,
i.e.,
${\rm SNR}_X = (\sum_{j=1}^{\infty}\lambda_j)^{1/2}/\sigma_e$
and ${\rm SNR}_Y= {\rm sd}(\int\beta X)/\sigma_{\varepsilon}$.
For simplicity,
we took ${\rm SNR}_X={\rm SNR}_Y$ ($=3$ or $10$).
To embody the sparsity assumptions,
in each sample,
$X_i$ was observed only at $L_i$ ($\stackrel{\iid}{\sim}{\rm Unif}\{3,4,5,6\}$)
points uniformly selected from $[0,1]$.
In total there were six combinations of settings.
200 iid samples were generated for each of them.
We randomly reserved 20\% of the subjects in each sample for testing 
and used the remainder for training. 
After running through all samples,
we computed 200 values of relative integrated squared estimation error (ReISEE)
$$
    {\rm ReISEE} = \|\beta\|_2^{-2}\|\beta-\hat{\beta}\|_2^2.
$$
Since neither PACE nor PLEASS leads to consistent predictions,
it is better to evaluate the prediction quality via the coverage percentage (CP)
of CIs constructed for testing subjects,
viz. 
$$
    {\rm CP} = \sum_{i\in I_{{\rm test}}}\mathbbm{1}\left\{\eta(X_{i})\in\widehat{\rm CI}_{i}\right\}\Big/\Big.\#I_{{\rm test}},
$$
where $\widehat{\rm CI}_{i}$ is the asymptotic (95\%) CI for $\eta(X_{i})$,
and $I_{{\rm test}}$ is the index set for testing portion
with cardinality $\#I_{{\rm test}}$.

\begin{figure}[!tp]
	\centering
	\begin{subfigure}{.5\textwidth}
		\centering
		\includegraphics[width=.9\textwidth, height=.25\textheight]
		    {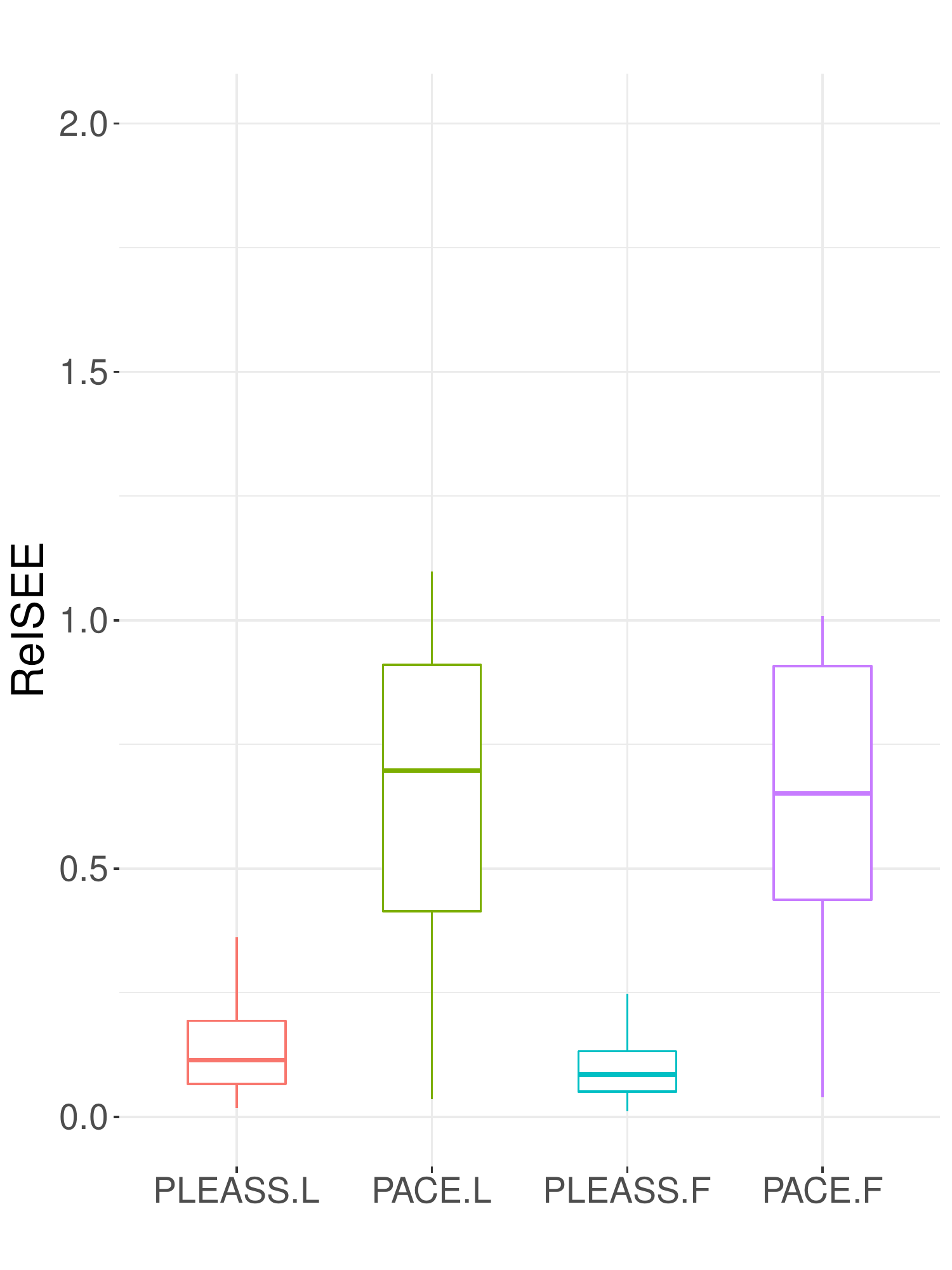}
		\caption{\footnotesize $\beta$ at \eqref{eq:beta.simu.1} with ${\rm SNR}_x={\rm SNR}_y=3$}
	\end{subfigure}%
	\begin{subfigure}{.5\textwidth}
		\centering
		\includegraphics[width=.9\textwidth, height=.25\textheight]
		    {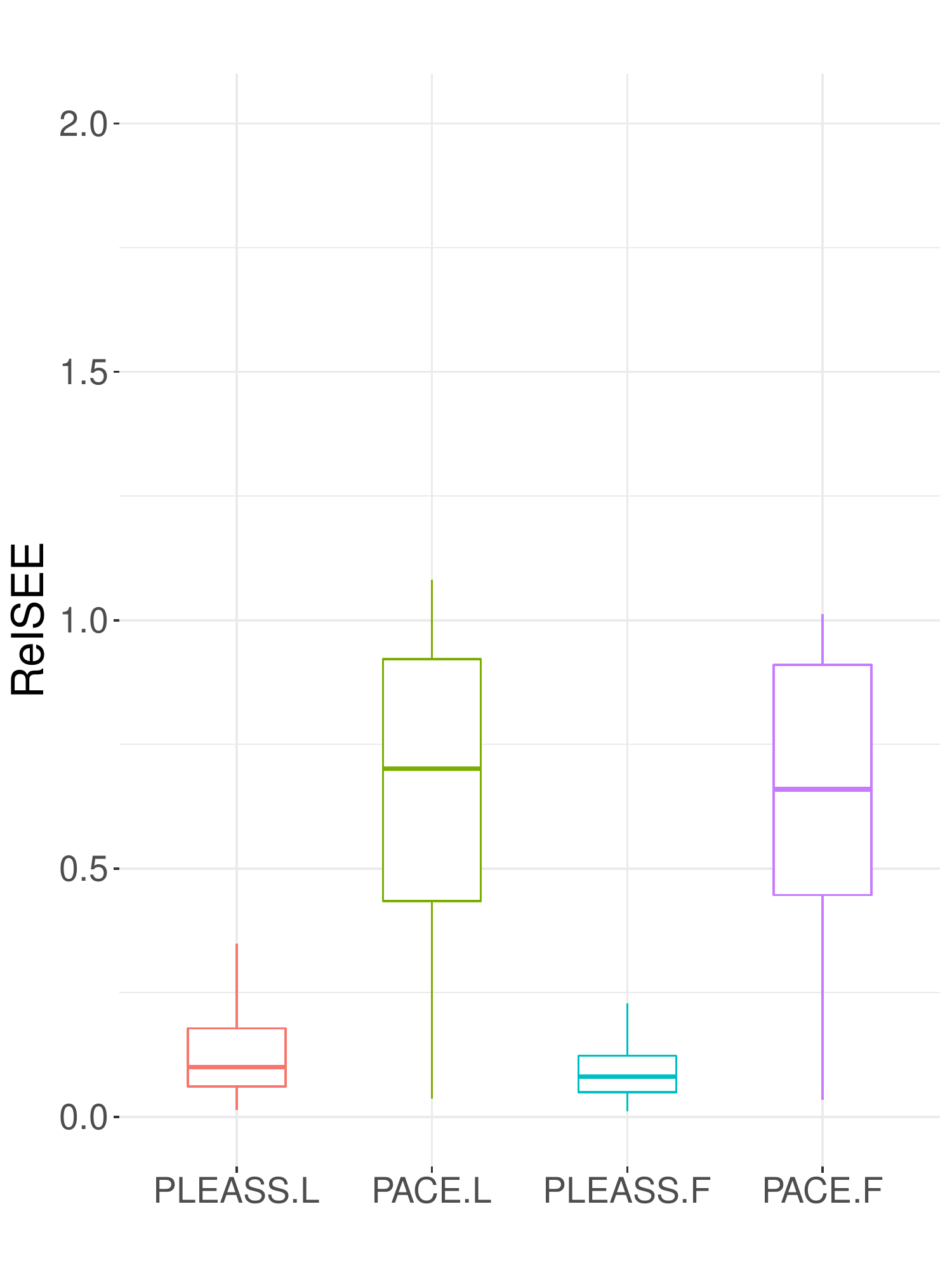}
		\caption{\footnotesize $\beta$ at \eqref{eq:beta.simu.1} with ${\rm SNR}_x={\rm SNR}_y=10$}
	\end{subfigure}
	\begin{subfigure}{.5\textwidth}
		\centering
		\includegraphics[width=.9\textwidth, height=.25\textheight]
		    {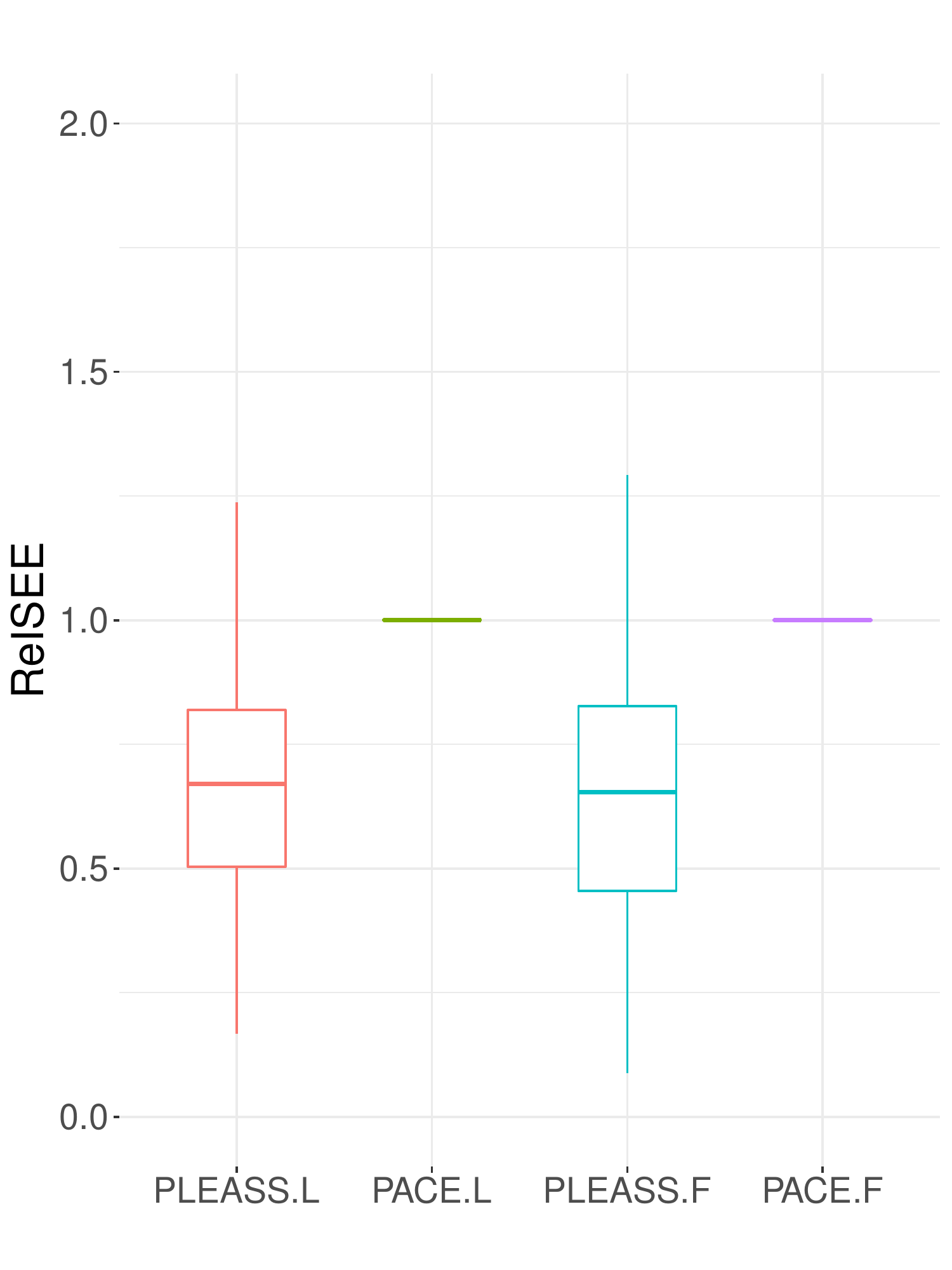}
		\caption{\footnotesize $\beta$ at \eqref{eq:beta.simu.2} with ${\rm SNR}_x={\rm SNR}_y=3$}
	\end{subfigure}%
	\begin{subfigure}{.5\textwidth}
		\centering
		\includegraphics[width=.9\textwidth, height=.25\textheight]
		    {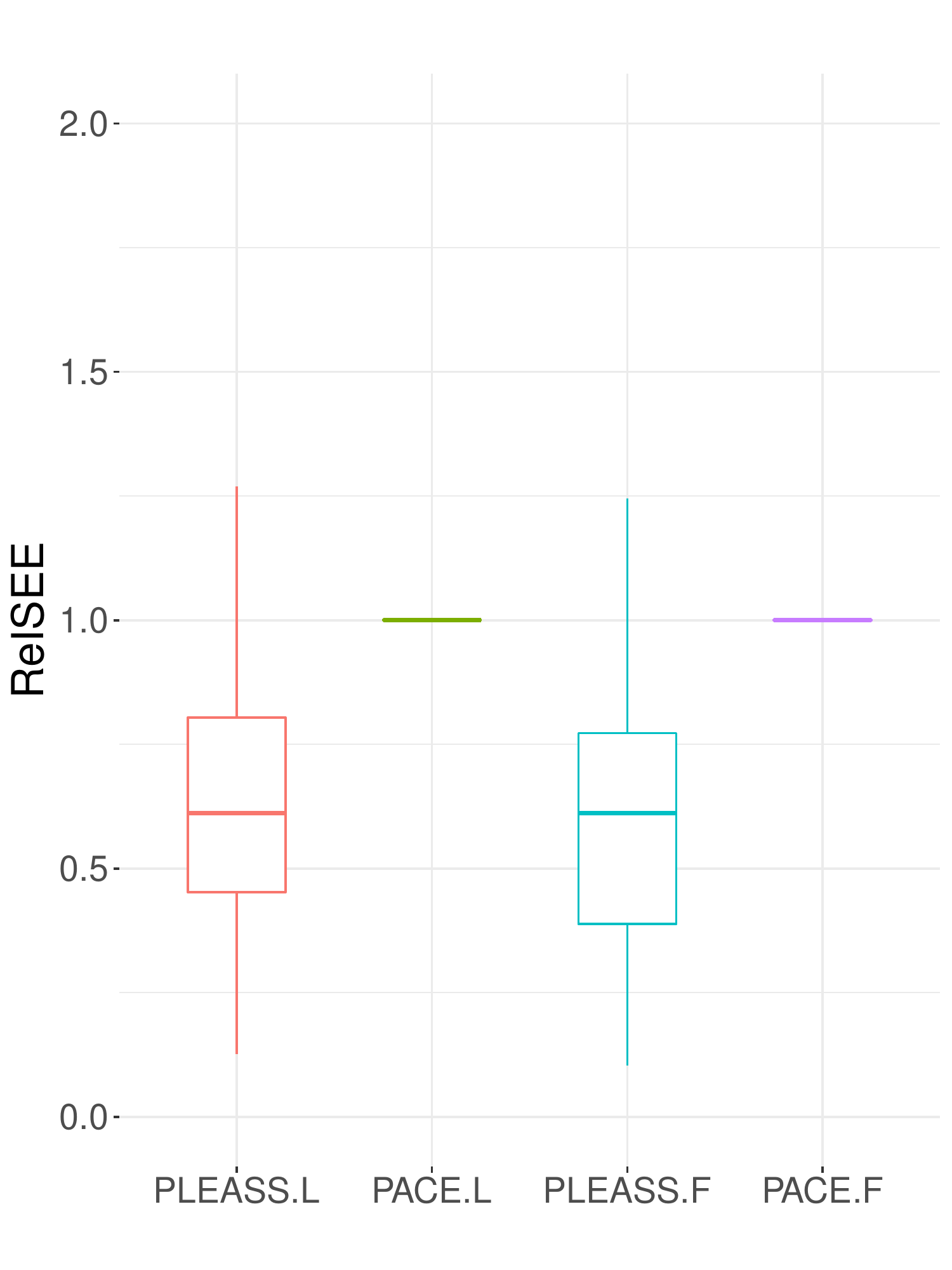}
		\caption{\footnotesize $\beta$ at \eqref{eq:beta.simu.2} with ${\rm SNR}_x={\rm SNR}_y=10$}
	\end{subfigure}
	\begin{subfigure}{.5\textwidth}
		\centering
		\includegraphics[width=.9\textwidth, height=.25\textheight]
		    {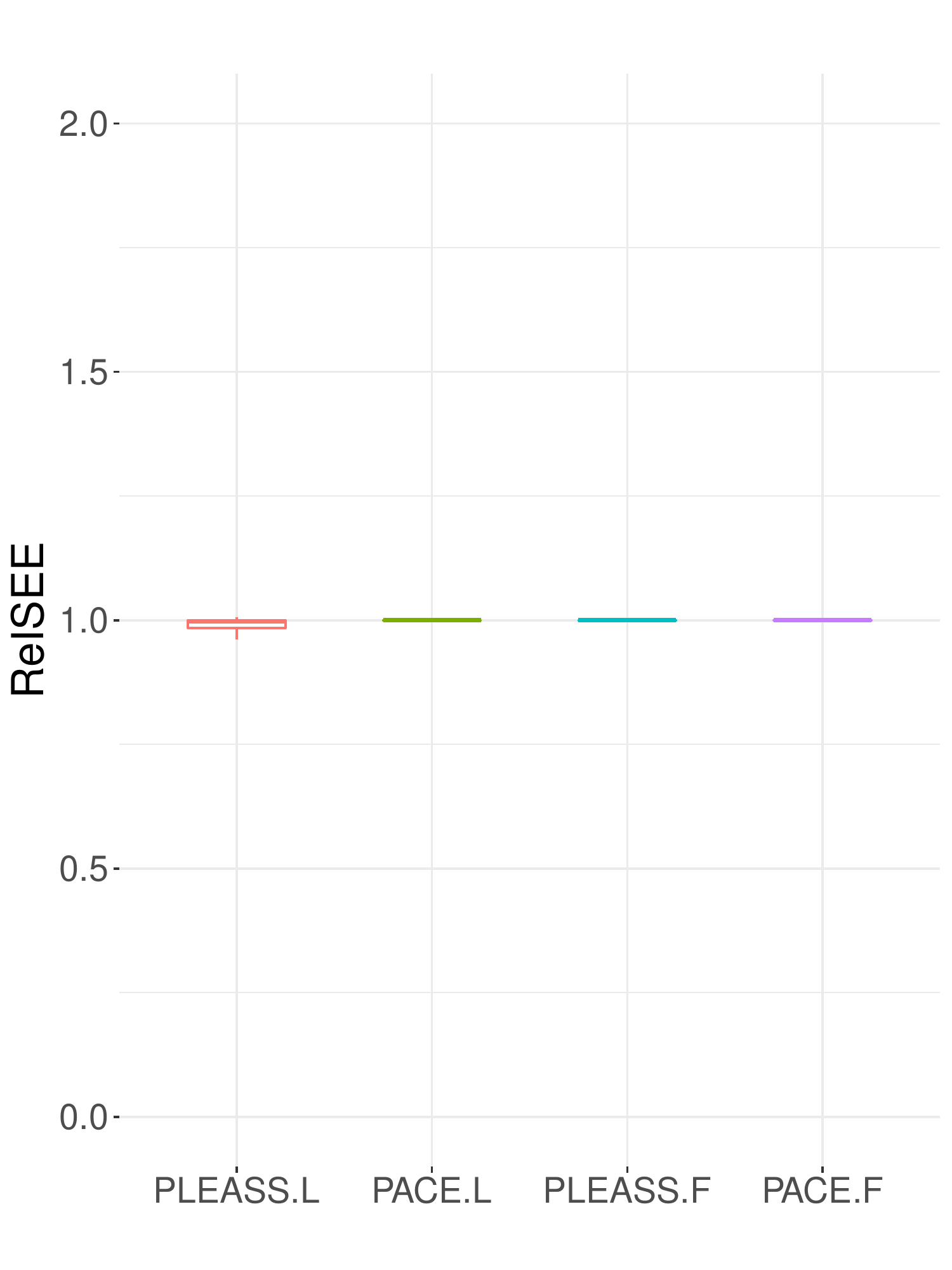}
		\caption{\footnotesize $\beta$ at \eqref{eq:beta.simu.3} with ${\rm SNR}_x={\rm SNR}_y=3$}
	\end{subfigure}%
	\begin{subfigure}{.5\textwidth}
		\centering
		\includegraphics[width=.9\textwidth, height=.25\textheight]
		    {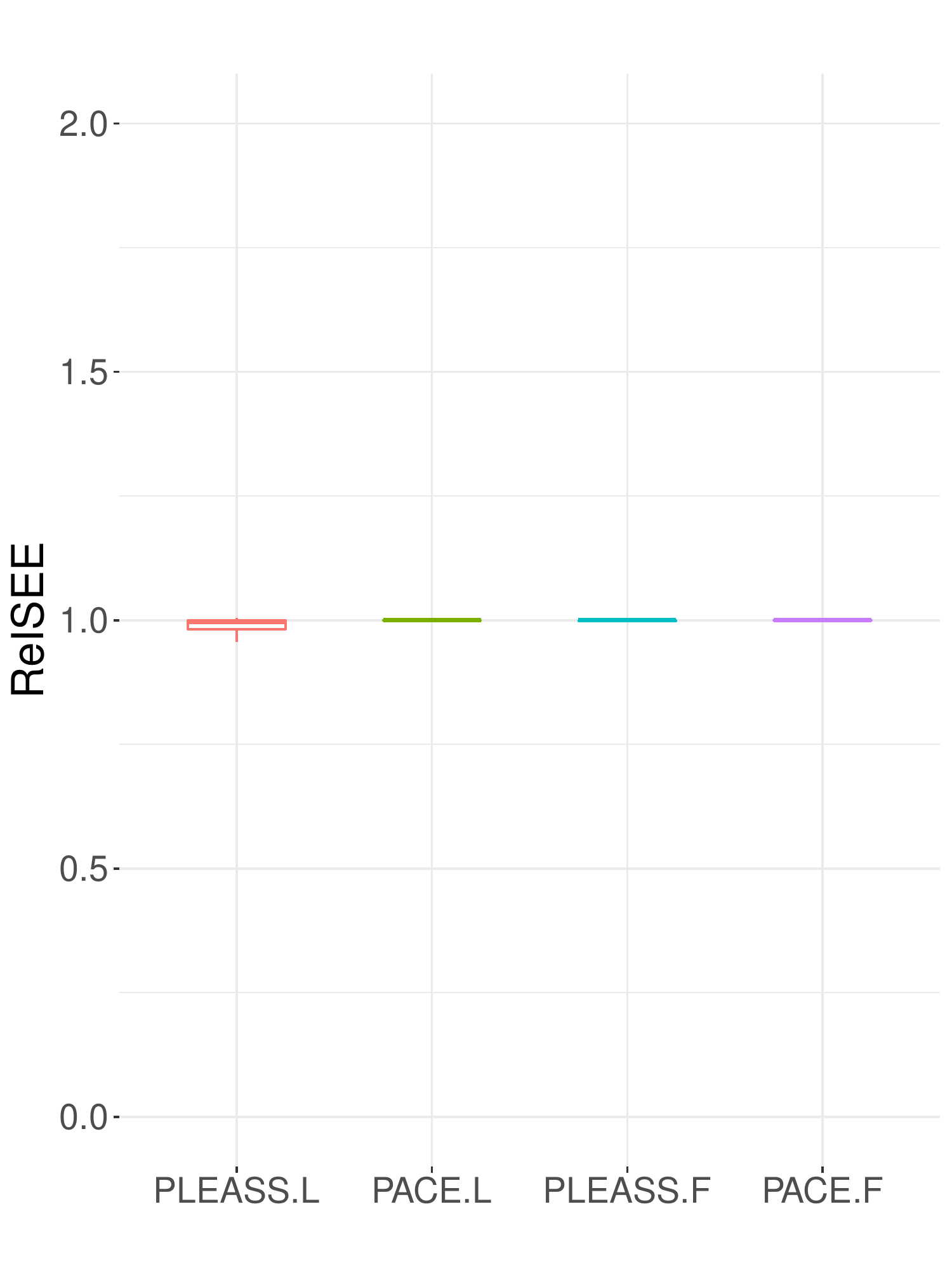}
		\caption{\footnotesize $\beta$ at \eqref{eq:beta.simu.3} with ${\rm SNR}_x={\rm SNR}_y=10$}
	\end{subfigure}
	\caption{\small
		Boxplots of ReISEE values under different simulated settings:
		SNR value varies with column,
        while rows differ in $\beta$.
		In each subfigure,
		from left to right,
		the four boxes respectively correspond to PLEASS.L, PLEASS.F, PACE.L, and PACE.F. 
	} 
	\label{fig:reisee.n=300}
\end{figure}

When $\beta$ was constructed from eigenfunctions corresponding to large or moderate eigenvalues
(viz. $\beta$ at \eqref{eq:beta.simu.1} or \eqref{eq:beta.simu.2}), 
PLEASS performed better in term of ReISEE;
see the first two rows of Figure \ref{fig:reisee.n=300}.
Particularly,
at the second row of Figure \ref{fig:reisee.n=300},
ReISEE values of PLEASS were mostly lower than one,
while PACE boxes was trapped at one.
An ReISEE box sticking around one
implied estimates concentrated around the most trivial $\hat{\beta}=0$,
i.e., the corresponding method failed to output non-trivial estimates.
This failure was caused by 
zero inner products between $\hat{v}_C$ and estimated basis functions;
this happened frequently
if $\beta$ was mainly associated with a small portion of total variation (of $\mathcal{V}_A$)
that was likely to be smoothed out in recovering $v_A$ and $v_C$.
Such was exactly the case for PACE in the scenario \eqref{eq:beta.simu.2}
and for both PACE and PLEASS with $\beta$ at \eqref{eq:beta.simu.3}.

\begin{figure}[!tp]
	\centering
	\begin{subfigure}{.5\textwidth}
		\centering
		\includegraphics[width=.9\textwidth, height=.25\textheight]
		    {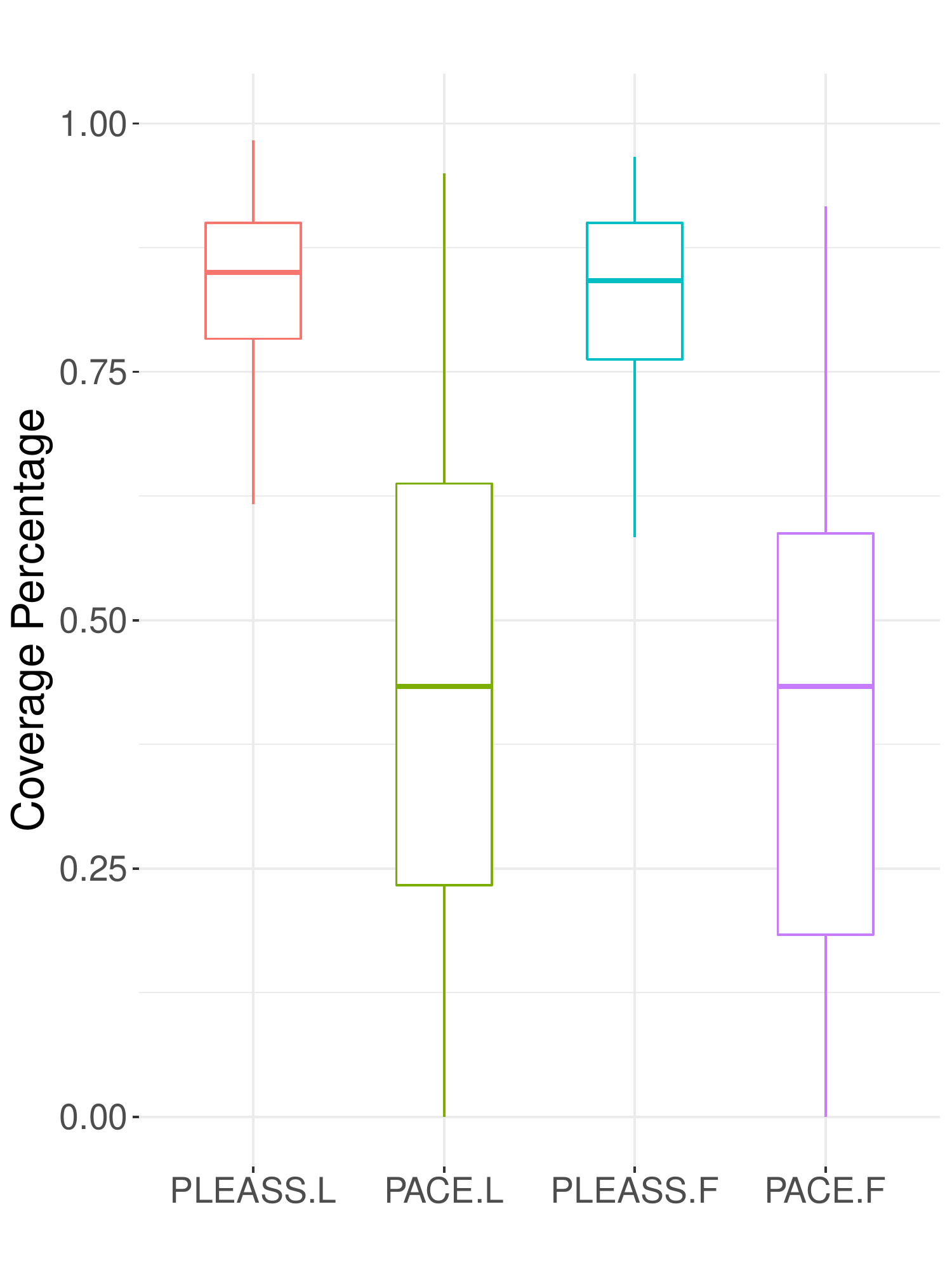}
		\caption{\footnotesize $\beta$ at \eqref{eq:beta.simu.1} with ${\rm SNR}_x={\rm SNR}_y=3$}
	\end{subfigure}%
	\begin{subfigure}{.5\textwidth}
		\centering
		\includegraphics[width=.9\textwidth, height=.25\textheight]
		    {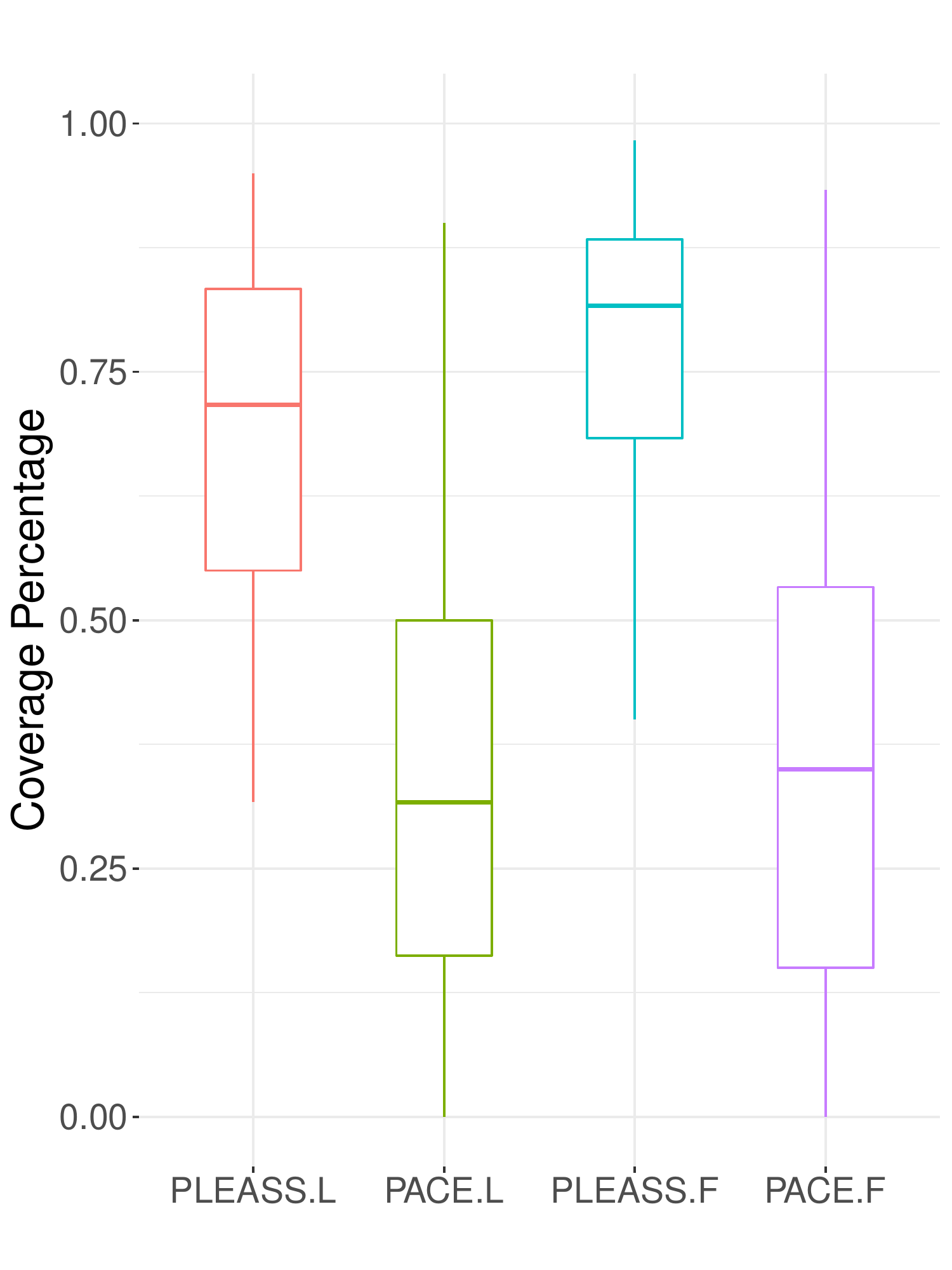}
		\caption{\footnotesize $\beta$ at \eqref{eq:beta.simu.1} with ${\rm SNR}_x={\rm SNR}_y=10$}
	\end{subfigure}
	\begin{subfigure}{.5\textwidth}
		\centering
		\includegraphics[width=.9\textwidth, height=.25\textheight]
		    {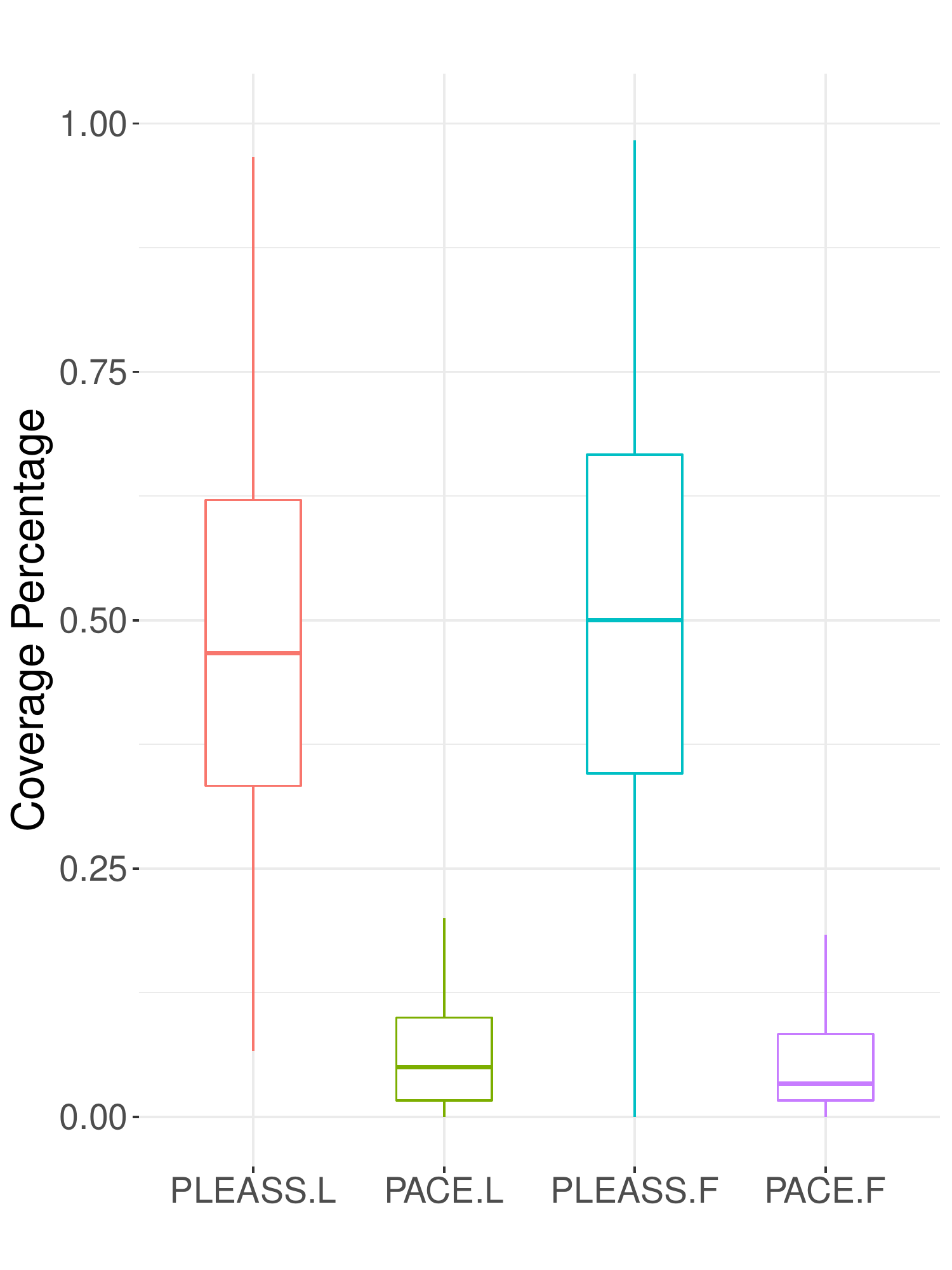}
		\caption{\footnotesize $\beta$ at \eqref{eq:beta.simu.2} with ${\rm SNR}_x={\rm SNR}_y=3$}
	\end{subfigure}%
	\begin{subfigure}{.5\textwidth}
		\centering
		\includegraphics[width=.9\textwidth, height=.25\textheight]
		    {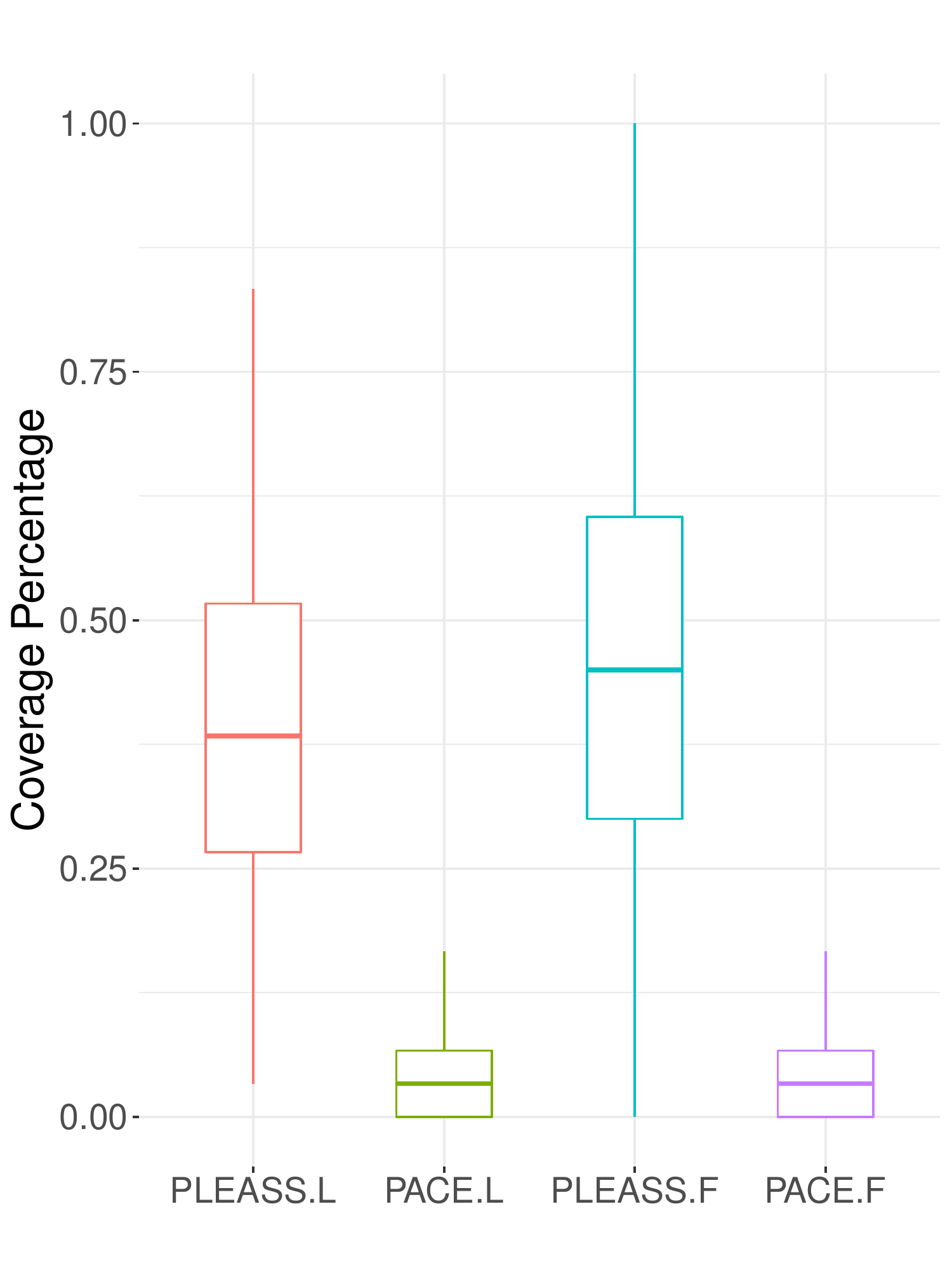}
		\caption{\footnotesize $\beta$ at \eqref{eq:beta.simu.2} with ${\rm SNR}_x={\rm SNR}_y=10$}
	\end{subfigure}
	\begin{subfigure}{.5\textwidth}
		\centering
		\includegraphics[width=.9\textwidth, height=.25\textheight]
		    {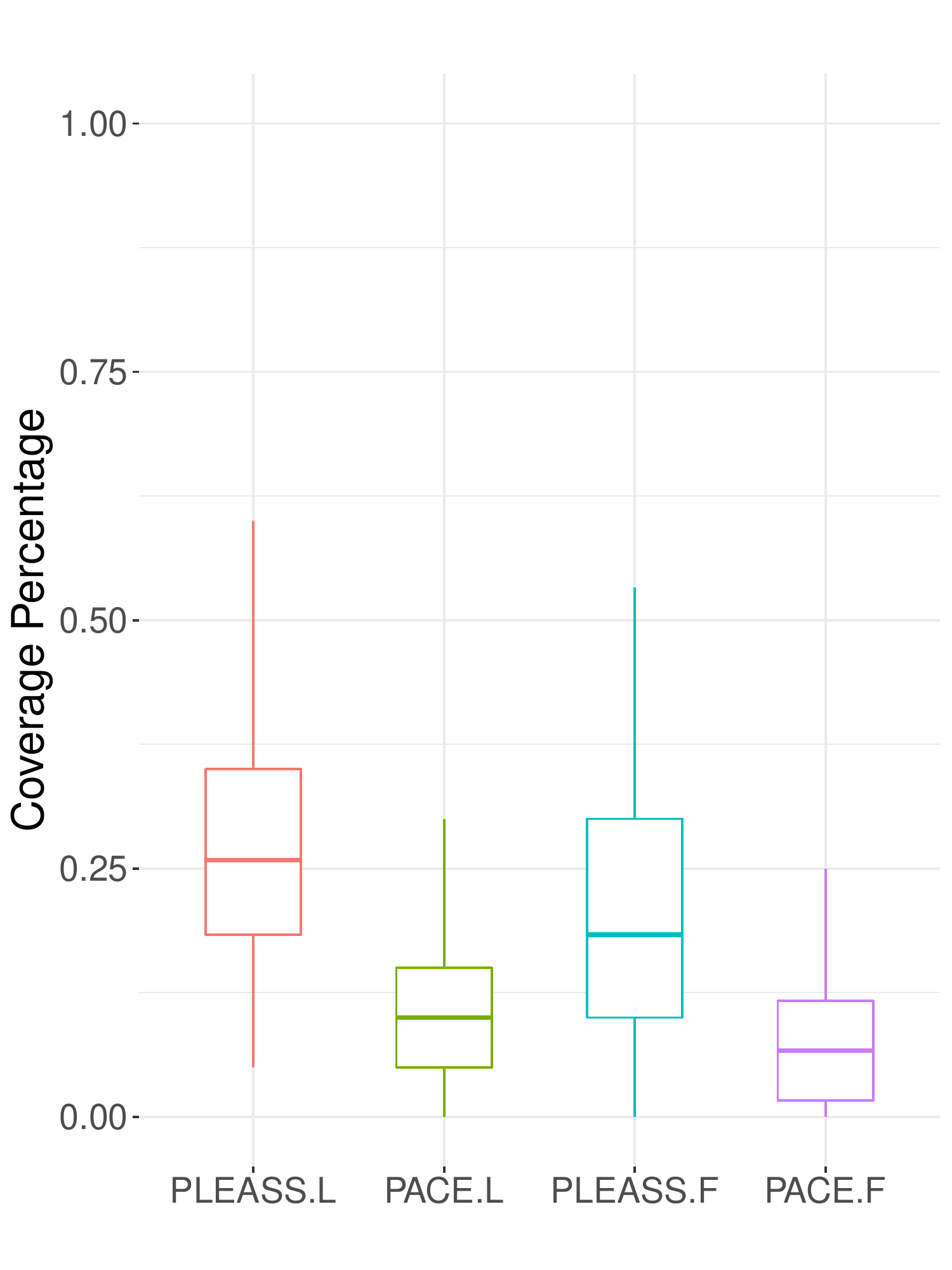}
		\caption{\footnotesize $\beta$ at \eqref{eq:beta.simu.3} with ${\rm SNR}_x={\rm SNR}_y=3$}
	\end{subfigure}%
	\begin{subfigure}{.5\textwidth}
		\centering
		\includegraphics[width=.9\textwidth, height=.25\textheight]
		    {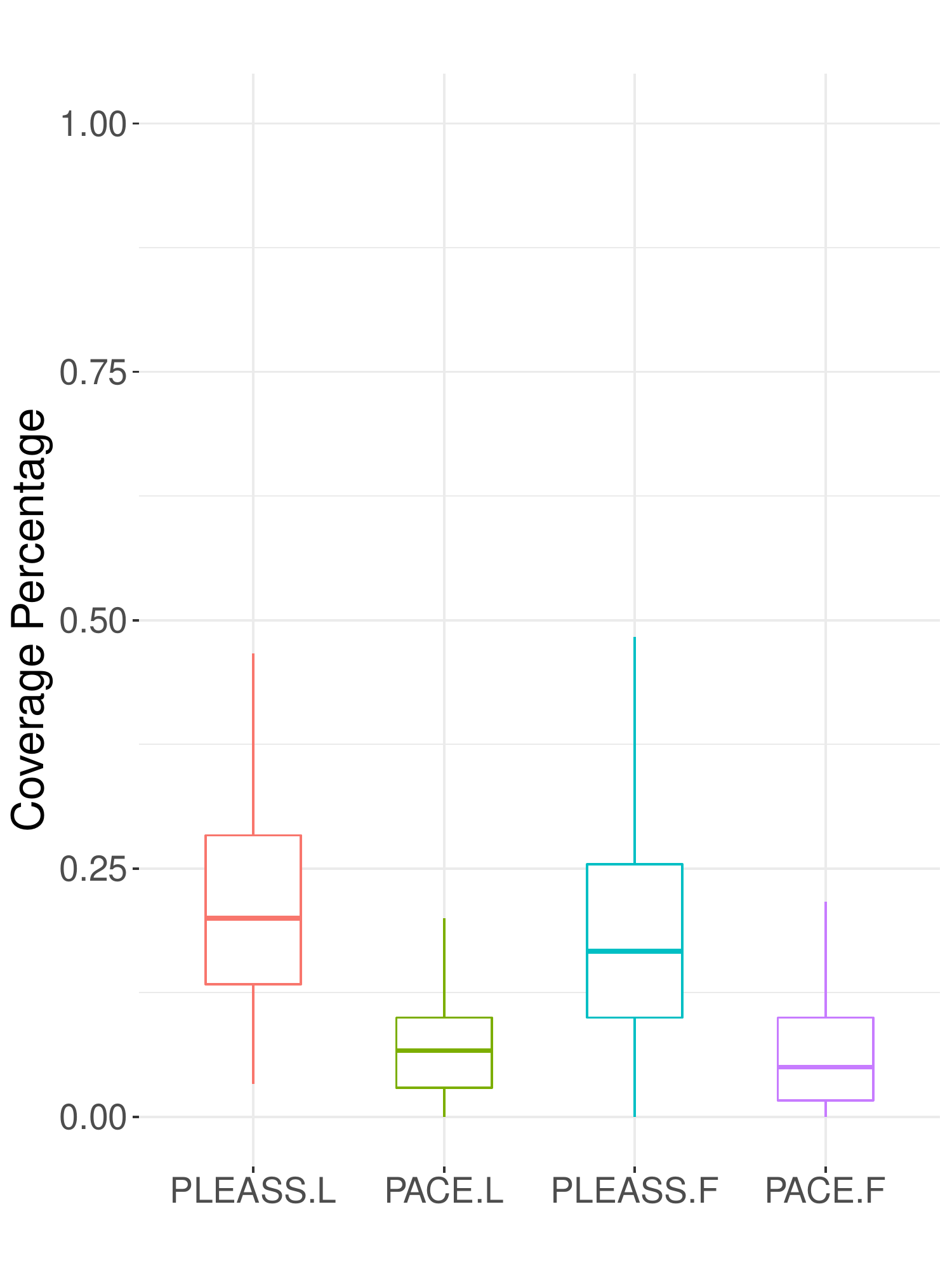}
		\caption{\footnotesize $\beta$ at \eqref{eq:beta.simu.3} with ${\rm SNR}_x={\rm SNR}_y=10$}
	\end{subfigure}
	\caption{\small
		Boxplots of CP values under different simulated settings:
		SNR value varies with column,
        while rows differ in $\beta$.
		In each subfigure,
		from left to right,
		the four boxes respectively correspond to PLEASS.L, PLEASS.F, PACE.L, and PACE.F. 
	} 
	\label{fig:cover.n=300}
\end{figure}

As seen in Figure \ref{fig:cover.n=300}, CP boxes belonging to PACE stayed at a low level, 
especially for scenarios \eqref{eq:beta.simu.2} and \eqref{eq:beta.simu.3}.
This phenomenon was consistent with the performance of PACE 
in estimating $\beta$ under corresponding settings.
In contrast,
PLEASS was more likely to output CP values closer to the stated level (95\%),
though we must admit that their coverages were still far from satisfactory 
especially with $\beta$ at \eqref{eq:beta.simu.2} and \eqref{eq:beta.simu.3}.
Looking into those $\eta(X_i)$ not covered by $\widehat{\rm CI}_i$,
we noticed that 
the majority of missed $\eta(X_i)$ fell at the right-hand side of $\widehat{\rm CI}_i$.
A possible cause of miss-covering lay in the bias of estimates for means of $X$ and $Y$;
a larger size of training set might be helpful.
Moreover,
although SNR had little impact on estimation
(compare the two columns of Figure \ref{fig:reisee.n=300}),
CP values appeared to be higher with a smaller SNR 
(compare the two columns of Figure \ref{fig:cover.n=300}):
$\eta(X_i)$ did not vary with SNR, 
while larger $\hat{\sigma}_e^2$ 
(resulting from smaller SNR)
widened $\widehat{\rm CI}_i$ and enhanced the coverage of $\widehat{\rm CI}_i$.

\subsection{Application to real datasets}

We then applied PLEASS to two real datasets.
The first came from a clinical trial,
whereas the second was densely observed but recorded with missing values.

\begin{description}
    \item[Primary Biliary Cholangitis (PBC) data.]
        Initially shared by \citet{TherneauGrambsch2000}, the
        dataset \texttt{pbcseq} 
        (accessible in \textsf{R}-package \texttt{survival}, \citealp{R-survival}) 
        was collected in a randomized placebo controlled trial of D-penicillamine, 
        a drug designed for PBC.
        PBC is a chronic disease in which bile ducts in the liver are slowly destroyed;
        it can cause more serious problems including liver cancer.
        All the participants of the clinical trial 
        were supposed to revisit the Mayo Clinic at six months, one year, and annually 
        after their initial diagnoses.
        However,
        participants' actual visiting frequencies, 
        with an average of 6,
        varied among patients,
        ranging from 1 to 16. This led to sparse and irregular recordings.
        Although the clinical trial lasted from January 1974 through May 1984,
        to satisfy the prerequisites of LLS,
        we  included only measurements within the first 3000 days
        and kicked out subjects with fewer than two visits.
        At each visit,
        several body indexes were measured and recorded,
        including alkaline phosphatase (ALP, in U/L)
        and aspartate aminotransferase (AST, in U/mL),
        both evaluating the health condition of liver. 
        We focused on this pair of indicators
        and attempted to model a linear connection between
        participants' latest AST measurements (response) 
        and their ALP profiles (functional predictor).
    \item[Diffusion tensor imaging (DTI) data.]
        Fractional anisotropy (FA)
        is measured at a specific spot in the white matter in the brain,
        ranging from 0 to 1 and
        reflecting the fiber density, axonal diameter and myelination.
        Along a tract of interest,
        these values forms an FA tract profile.
        Collected at the Johns Hopkins University 
        and Kennedy-Krieger Institute,
        dataset \texttt{DTI} 
        (in \textsf{R}-package \texttt{refund}, \citealp{R-refund})
        contained FA tract profiles for the corpus callosum
        measured via DTI.
        Though these trajectories were not sparsely measured,
        a few of them suffered from missing records
        which could be handled by PACE and PLEASS 
        without presmoothing or interpolation.
        We investigated the relationship between 
        participants' FA tract profiles (predictor)
        and their Paced Auditory Serial Addition Test (PASAT) scores (response),
        where PASAT is a traditional tool assessing impairments in the cognitive functioning
        and is extensively used in the diagnosis of,
        e.g., the multiple sclerosis \citep{Tombaugh2006}.
\end{description}

For each dataset,
200 random splits were carried out.
In each split,
(roughly) 80\% of the subjects were put into the training set
while the remainder were kept for testing. 
After predicting responses for the test set,
we generated values of relative mean squared prediction error (ReMSPE), viz.
$$
    {\rm ReMSPE}=\frac{\sum_{i\in I_{{\rm test}}}(Y_{i}-\widehat{Y}_{i})^2}
        {\sum_{i\in I_{{\rm test}}}(Y_{i}-\bar{Y}_{{\rm train}})^2},
$$
for each approach and each split.
Here $\widehat{Y}_i$
is the prediction for the $i$th response,
and $\bar{Y}_{{\rm train}}$ is the mean training response.
ReMSPE values for PBC and DTI cases were collected 
and summarized into boxes;
see Figure \ref{fig:remspe}.
In both applications,
PLEASS was demonstrated to be more competitive than PACE,
enjoying lower medians and smaller dispersion of ReMSPE values.
Analogous to the previous simulation study,
Figure \ref{fig:remspe} shows that
FACE performs close to LLS
when used with PLEASS.
As a result,
PLEASS.F might be preferred if a low time consumption were particularly appreciated.


\begin{figure}[t!]
	\centering
	\begin{subfigure}{.5\textwidth}
		\centering
		\includegraphics[width=\textwidth, height=.3\textheight]
		    {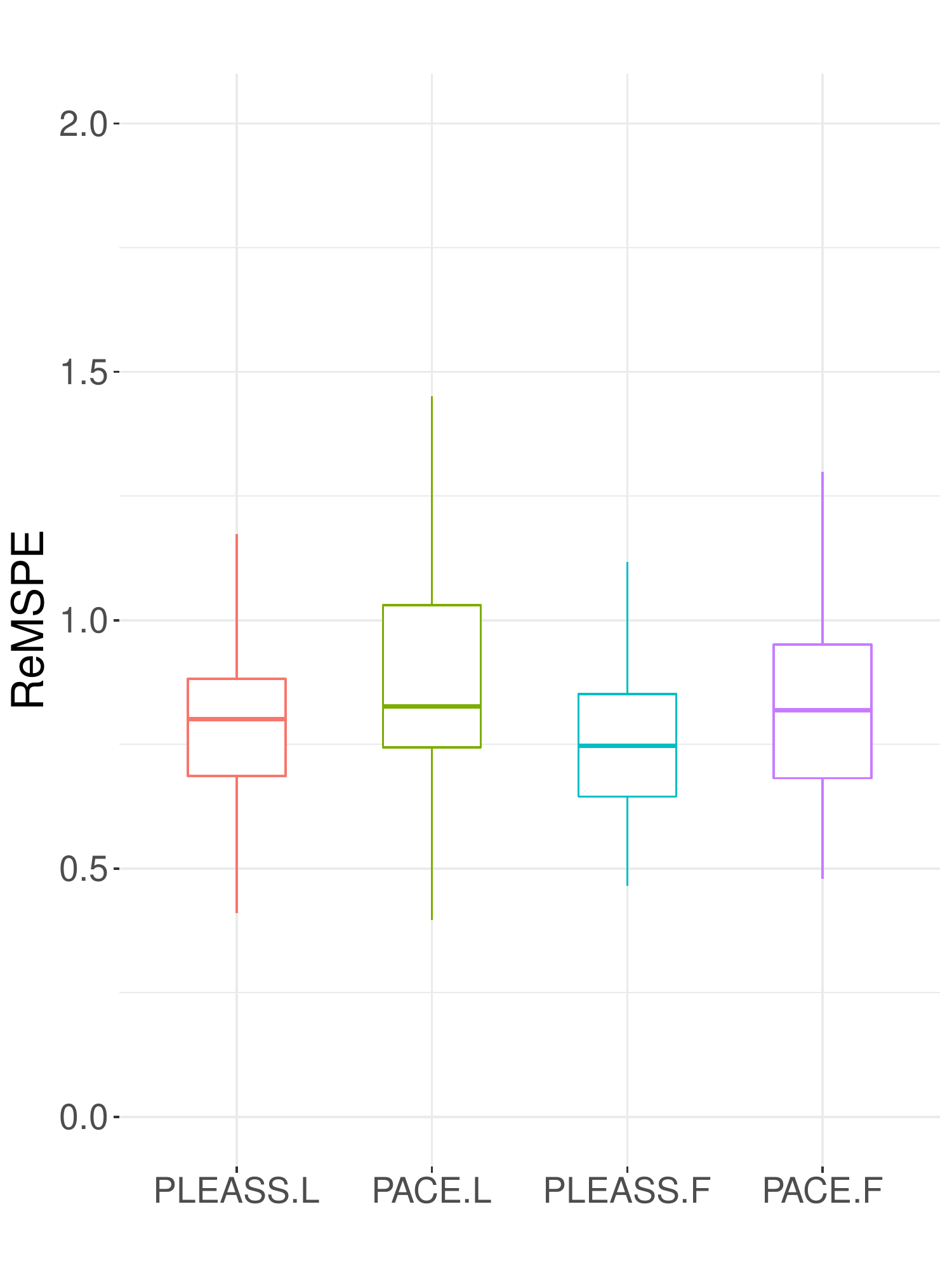}
		\caption{\footnotesize ALP (predictor) vs. lastest AST (response)}
	\end{subfigure}%
	\begin{subfigure}{.5\textwidth}
		\centering
		\includegraphics[width=\textwidth, height=.3\textheight]
		    {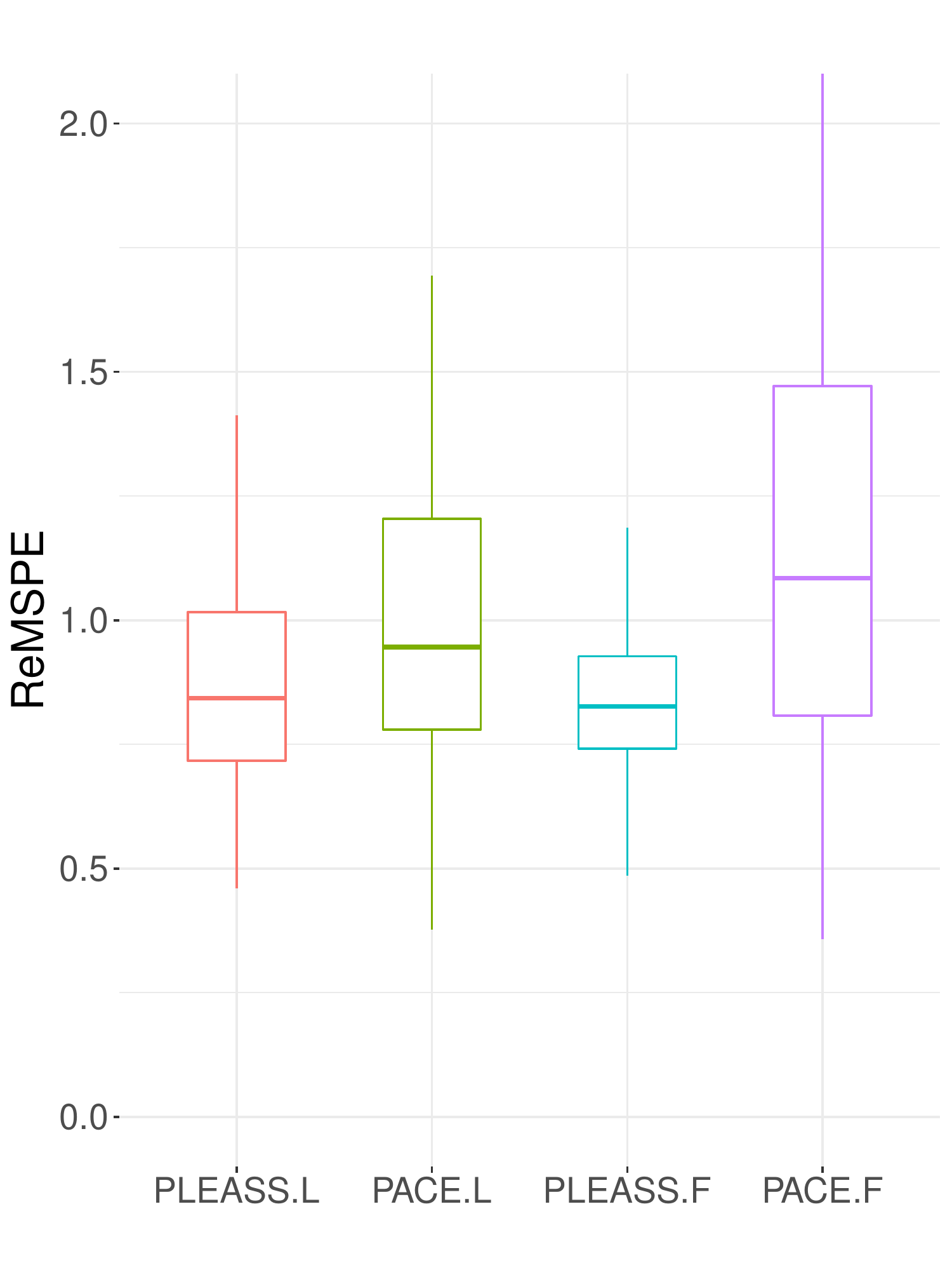}
		\caption{\footnotesize FA tract profile (predictor) vs. PASAT (response)}
	\end{subfigure}%
	\caption{\small 
		ReMSPE boxplots for real data analysis.
		In each subfigure,
		from left to right,
		the four boxes respectively correspond to PLEASS.L, PLEASS.F, PACE.L, and PACE.F. 
	} 
	\label{fig:remspe}
\end{figure}

\section{Conclusion and discussion}\label{sec:conclusion}

The main contributions of our work are summarized as follows.
First,
we propose PLEASS, 
a variant of FPLS modified for scenarios
in which functional predictors are observed sparsely and with contamination.
Second,
not only do we give
estimators and predictions via PLEASS, but
also we construct CIs for mean responses.
Allowing $p$ to diverge as a function of $n$,
our theoretical work
is among the few asymptotic results available for FPLS and its variants.
Third,
we numerically reveal the advantage of PLEASS 
in specific scenarios.


Estimators for the variance and covariance structure
may be further revised.
If trajectories are no longer independent of each other 
(e.g., spatially correlated curves representing distinct cities),
it is more reasonable to employ
the proposal of \citet{PaulPeng2011},
viz. a weighted version of LLS.
Another concern is the nature of missingness:
the mode of sparsity here is assumed independent across trajectories and measurement errors.
Even if the missingness is permitted to be correlated with the values of unobserved time points, 
we speculate that,
after necessary modifications, 
estimates of the ML type would be still promising in estimating components of covariance structure.

In contrast with PLEASS,
which is for now concentrated on SoFR only,
PACE is more versatile:
it is applicable even to function-on-function regression 
(FoFR, with response and predictor both functional) and 
is capable as well of recovering predictor trajectories.
Merging PLEASS into the framework of \cite{Zhou2020},
we may adapt it to FoFR with sparsely/densely observed functional predictors/responses.
Moreover the application of PLEASS is not limited to linear models,
since it is practicable to 
embed FPLS techniques into the iteratively reweighted least squares for maximizing likelihood
\citep{Marx1996};
\cite{Albaqshi2017} and \cite{WangIbrahimZhu2020}
successfully applied this idea to 
functional logistic regression and functional joint modeling, respectively.

\section*{Acknowledgment}

Special thanks go to Professors Ling Zhou and Huazhen Lin
(both serving for the Southwestern University of Finance and Economics, China)
and Professor Hua Liang
(George Washington University, United States)
for the generous sharing of their source codes.
The authors' work is financially supported by
the Natural Sciences and Engineering Research Council of Canada (NSERC).

\appendix

\section{Local linear smoother}\label{appendix:LLS}

Let $\kappa=\kappa(\cdot)$ be a function on $\mathbb{R}$ satisfying 
\ref{cond:kernel.1}--\ref{cond:kernel.3} in Appendix \ref{appendix:technical};
examples include the symmetric Beta family
\citep[][Eq.~2.5]{FanGijbels1996}
which has the Epanechnikov kernel 
$\kappa(t)=.75(1-t^2)\mathbbm{1}(|t|\leq 1)$ as a special case.
LLS actually falls into the framework of weighted least squares (WLS)
\citep[][pp.~58--59]{FanGijbels1996}.
Given integers $M$ and $m$ 
(with values specified in the following cases \ref{LLS:mu.x}--\ref{LLS:sigma.2.x}) and
matrices $\mathbf{1}_M$ (the $M$-vector of ones), 
$\bm{u}$ (an $M$-vector), 
$\mathbf{T}$ (an $M\times m$ matrix)
and $\mathbf{W}$ (an $M\times M$ non-negative definite matrix),
one solves
\[
    \min_{a_0,\bm{a}}(\bm{u}-a_0\mathbf{1}_M-\mathbf{T}\bm{a})^\top
    \mathbf{W}(\bm{u}-a_0\mathbf{1}_M-\mathbf{T}\bm{a})
\]
for a scalar $a_0\in\mathbb{R}$ and an $m$-vector $\bm{a}=[a_1,\ldots, a_m]^\top$.
In fact, LLS  only uses the WLS solution for $a_0$,
namely,
\begin{align}
    \hat{a}_0
    &=(
        \mathbf{1}_M^\top\mathbf{W}^{1/2}\mathbf{P}_{\mathbf{W}^{1/2}\mathbf{T}}^{\perp}
        \mathbf{W}^{1/2}\mathbf{1}_M
    )^+
    \mathbf{1}_M^\top\mathbf{W}^{1/2}\mathbf{P}_{\mathbf{W}^{1/2}\mathbf{T}}^{\perp}
    \mathbf{W}^{1/2}\bm{u}
    \notag\\
    &=[
        \mathbf{1}_M^\top
        \{\mathbf{W}-
            \mathbf{W}\mathbf{T}
            (\mathbf{T}^\top\mathbf{W}\mathbf{T})^+
            \mathbf{T}^\top\mathbf{W}
        \}
        \mathbf{1}_M
    ]^+
    \mathbf{1}_M^\top
    \{\mathbf{W}-
        \mathbf{W}\mathbf{T}
        (\mathbf{T}^\top\mathbf{W}\mathbf{T})^+
        \mathbf{T}^\top\mathbf{W}
    \}
    \bm{u}
    \label{eq:a0.hat}
\end{align}
in which  the Moore-Penrose generalized inverse is denoted by ``$+$''
and 
$
    \mathbf{P}_{\mathbf{W}^{1/2}\mathbf{T}}^{\perp}
    =\mathbf{I}-
        \mathbf{W}^{1/2}\mathbf{T}
        (\mathbf{T}^\top\mathbf{W}\mathbf{T})^+
        \mathbf{T}^\top\mathbf{W}^{1/2}.
$
In particular,
four different combinations of $\bm{u}$, $\mathbf{T}$ and $\mathbf{W}$ yield
estimates of the four targets of interest,
$\mu_X$, $v_C$, $v_A$, and $\tilde{v}$, as follows:
\begin{enumerate}[label=(\roman*)] 
    \item\label{LLS:mu.x}
        Given $t\in\mathbb{T}$,
        estimate $\mu_X(t)$ by $\hat{\mu}_X(t)=\hat{a}_0$ from \eqref{eq:a0.hat}
        with $\sum_{1\leq i\leq n}L_i$-vectors
        \[
            \bm{u}=
            \left[
                \widetilde{X}_1(T_{11})
                ,\ldots, 
                \widetilde{X}_1(T_{1L_1})
                ,\ldots, 
                \widetilde{X}_n(T_{n1})
                ,\ldots, 
                \widetilde{X}_n(T_{n L_n})
            \right]^\top
        \]
        and
        $
            \mathbf{T}=
            \left[
                t-T_{11}, \ldots, t-T_{1L_1}, \ldots, t-T_{n1}, \ldots, t-T_{n L_n}
            \right]^\top
        $
        and $\sum_{1\leq i\leq n}L_i\times \sum_{1\leq i\leq n}L_i$ matrix
        \[
            \mathbf{W}=\text{diag}
            \left\{
                \kappa\left(\frac{t-T_{11}}{h_{\mu}}\right), 
                \ldots, 
                \kappa\left(\frac{t-T_{1L_1}}{h_{\mu}}\right), 
                \ldots, 
                \kappa\left(\frac{t-T_{n1}}{h_{\mu}}\right), 
                \ldots, 
                \kappa\left(\frac{t-T_{nL_n}}{h_{\mu}}\right) 
            \right\}.
        \]
    \item 
        Write $\bar{Y}=n^{-1}\sum_{1\leq i\leq n}Y_i$.
        For arbitrary $t\in\mathbb{T}$,
        $\hat{v}_C(t)=\hat{a}_0-\bar{Y}\cdot\hat{\mu}_X(t)$,
        where $\hat{a}_0$ follows \eqref{eq:a0.hat}
        with $\sum_{1\leq i\leq n}L_i$-vectors
        \begin{multline*}
            \bm{u}=
            \Big[
                \{\widetilde{X}_1(T_{11})-\hat{\mu}_X(T_{11})\}(Y_1-\bar{Y}),
                \ldots,
                \{\widetilde{X}_1(T_{1L_1})-\hat{\mu}_X(T_{1L_1})\}(Y_1-\bar{Y}),
                \ldots, 
                \\
                \{\widetilde{X}_n(T_{n1})-\hat{\mu}_X(T_{n1})\}(Y_n-\bar{Y}),
                \ldots, 
                \{\widetilde{X}_n(T_{n L_n})-\hat{\mu}_X(T_{n L_n})\}
                    (Y_n-\bar{Y})
            \Big]^\top
        \end{multline*}
        and
        $
            \mathbf{T}=
            \left[
                t-T_{11}, \ldots, t-T_{1L_1}, \ldots, t-T_{n1}, \ldots, t-T_{n L_n}
            \right]^\top
        $
        as well as $\sum_{1\leq i\leq n}L_i\times \sum_{1\leq i\leq n}L_i$ matrix
        \[
            \mathbf{W}=\text{diag}
            \left\{
                \kappa\left(\frac{t-T_{11}}{h_C}\right), 
                \ldots, 
                \kappa\left(\frac{t-T_{1L_1}}{h_C}\right), 
                \ldots, 
                \kappa\left(\frac{t-T_{n1}}{h_C}\right), 
                \ldots, 
                \kappa\left(\frac{t-T_{nL_n}}{h_C}\right) 
            \right\}.
        \]
    \item 
        Fix $s,t\in\mathbb{T}$. 
        Then $\hat{v}_A(s,t)=\hat{a}_0-\hat{\mu}_X(s)\hat{\mu}_X(t)$,
        where $\hat{a}_0$ is fitted as \eqref{eq:a0.hat}
        with $\sum_{1\leq i\leq n}L_i(L_i-1)$-vector
        \begin{multline*}
            \bm{u}=
            \Big[
                \ldots,
                \widetilde{X}_i(T_{i\ell})\widetilde{X}_i(T_{i1}),
                \ldots,
                \widetilde{X}_i(T_{i\ell})\widetilde{X}_i(T_{i,\ell-1}),
                \\
                \widetilde{X}_i(T_{i\ell})\widetilde{X}_i(T_{i,\ell+1}),
                \ldots,
                \widetilde{X}_i(T_{i\ell})\widetilde{X}_i(T_{i L_i})
                \ldots
            \Big]^\top,  
        \end{multline*}
        $\sum_{1\leq i\leq n}L_i(L_i-1)\times 2$ matrix
        \[
            \mathbf{T}=
            \left[\begin{array}{cccccccc}
                \ldots &s-T_{i\ell} &\cdots &s-T_{i\ell} &s-T_{i\ell} &\cdots &s-T_{i\ell} &\cdots \\
                \cdots &t-T_{i1} &\cdots &t-T_{i,\ell-1} &t-T_{i,\ell+1} &\cdots &t-T_{i L_i} &\cdots
            \end{array}\right]^\top
        \]
        and $\sum_{1\leq i\leq n}L_i(L_i-1)\times\sum_{1\leq i\leq n}L_i(L_i-1)$ matrix
        \begin{multline*}
            \mathbf{W}=\text{diag}
            \Bigg\{
                \ldots, 
            \Bigg.
            \Bigg.
                \kappa\left(\frac{s-T_{i\ell}}{h_A}\right)
                    \kappa\left(\frac{t-T_{i1}}{h_A}\right), 
                \ldots, 
                \kappa\left(\frac{s-T_{i\ell}}{h_A}\right)
                    \kappa\left(\frac{t-T_{i,\ell-1}}{h_A}\right),
            \Bigg.\\
            \Bigg.
                \kappa\left(\frac{s-T_{i\ell}}{h_A}\right)
                    \kappa\left(\frac{t-T_{i,\ell+1}}{h_A}\right),
                \ldots, 
                \kappa\left(\frac{s-T_{i\ell}}{h_A}\right)
                    \kappa\left(\frac{t-T_{iL_i}}{h_A}\right), 
                \ldots
            \Bigg\}.    
        \end{multline*}
    \item\label{LLS:sigma.2.x}
        Rotate the two-tuple $(T_{i\ell_1}, T_{i\ell_2})$ to become
        $$
            \left[
                \begin{array}{c}
                    T_{i\ell_1}^\#  \\
                    T_{i\ell_2}^\#
                \end{array}
            \right]
            =
            \left[
                \begin{array}{cc}
                    \sqrt{2}/2  &\sqrt{2}/2  \\
                    -\sqrt{2}/2 &\sqrt{2}/2 
                \end{array}
            \right]
            \left[
                \begin{array}{c}
                    T_{i\ell_1}  \\
                    T_{i\ell_2}
                \end{array}
            \right].
        $$
        For arbitrarily fixed $t\in\mathbb{T}$,
        $\tilde{v}(t)=\hat{a}_0-\hat{\mu}_X^2(t)$,
        where $\hat{a}_0$ follows \eqref{eq:a0.hat}
        with $\sum_{1\leq i\leq n}L_i$-vector
        $$
            \bm{u}=
            \left[
                \widetilde{X}_1^2(T_{11}^\#)
                \ldots 
                \widetilde{X}_1^2(T_{1L_1}^\#)
                \ldots 
                \widetilde{X}_n^2(T_{n1}^\#)
                \ldots 
                \widetilde{X}_n^2(T_{n L_n}^\#)
            \right]^\top, 
        $$
        $\sum_{1\leq i\leq n}L_i\times 2$ matrix
        $$
            \mathbf{T}=
            \left[\begin{array}{ccccccc}
                -T_{11}^\# &\cdots &-T_{1L_i}^\# &\cdots &-T_{n1}^\# &\cdots &-T_{nL_i}^\#
                \\
                t/\sqrt{2}-T_{11}^\# &\cdots &t/\sqrt{2}-T_{1L_1}^\# &\cdots 
                    &t/\sqrt{2}-T_{n1}^\# &\cdots &t/\sqrt{2}-T_{nL_n}^\#
            \end{array}\right]^\top
        $$
        and $\sum_{1\leq i\leq n}L_i\times\sum_{1\leq i\leq n}L_i$ matrix
        \begin{multline*}
            \mathbf{W}=\text{diag}
            \Bigg\{
                \kappa\left(\frac{t/\sqrt{2}-T_{11}^\#}{h_{\sigma}}\right), 
                \ldots, 
                \kappa\left(\frac{t/\sqrt{2}-T_{1L_1}^\#}{h_{\sigma}}\right), 
                \ldots,
                \\
                \kappa\left(\frac{t/\sqrt{2}-T_{n1}^\#}{h_{\sigma}}\right), 
                \ldots, 
                \kappa\left(\frac{t/\sqrt{2}-T_{nL_n}^\#}{h_{\sigma}}\right) 
            \Bigg\}.
        \end{multline*}
        Then, 
        as suggested in \citet{YaoMullerWang2005a, YaoMullerWang2005b},
        $\sigma_e^2$ is estimated by averaging 
        $\tilde{v}(t)-\hat{v}_A(t,t)$ over a truncated version of $\mathbb{T}=[0,1]$, 
        say
        $
            \mathbb{T}_1=[1/4, 3/4],
        $
        i.e.,
        $
            \hat{\sigma}_e^2
            =2\int_{\mathbb{T}_1}\{\tilde{v}(t)-\hat{v}_A(t,t)\}\dd t.
        $
\end{enumerate}

Bandwidths $h_{\mu}$, $h_C$, $h_A$ and $h_{\sigma}$ 
are all tuned through GCV,
i.e., they are chosen to minimize
\begin{multline*}
    \frac{\bm{u}^\top\mathbf{W}^{1/2}
        \mathbf{P}_{\mathbf{W}^{1/2}[\bm{1}_M,\mathbf{T}]}^{\perp}
        \mathbf{W}^{1/2}\bm{u}}
        {\{\sum_{i=1}^nL_i-\text{tr}(\mathbf{P}_{\mathbf{W}^{1/2}[\bm{1}_M,\mathbf{T}]})\}^2}
    \\
    =
    \frac{\bm{u}^\top\{
            \mathbf{W}-
                \mathbf{W}[\bm{1}_M,\mathbf{T}]
                ([\bm{1}_M,\mathbf{T}]^\top\mathbf{W}[\bm{1}_M,\mathbf{T}])^+
                [\bm{1}_M,\mathbf{T}]^\top\mathbf{W}
        \}\bm{u}}
        {\{\sum_{i=1}^nL_i-\text{tr}(\mathbf{P}_{\mathbf{W}^{1/2}[\bm{1}_M,\mathbf{T}]})\}^2}
\end{multline*}
with their respective corresponding $\bm{u}$, $\mathbf{T}$ and $\mathbf{W}$.
\citet[][Eq.~4.3]{FanGijbels1996} suggested a rule of thumb
which is a good starting point in determining candidate pools for bandwidths.

\section{Technical details: assumptions, lemmas, and proofs}\label{appendix:technical}

Recall the setting of sparsity and error-in-variable:
for the $i$th subject,
given the number of observation times $L_i\stackrel{\iid}{\sim} L$ (satisfying \ref{cond:rv.1}), 
noisy trajectories $\widetilde{X}_i$ are observed only at time points 
$T_{i\ell}\stackrel{\iid}{\sim} T$
such that $\widetilde{X}_i(T_{i\ell})=X_i(T_{i\ell})+\sigma_e e_{i\ell}$,
$\ell = 1,\ldots, L_i$,
where $X_i$ ($\stackrel{\iid}{\sim} X$) are underlying functional predictors,
and measurement errors $e_{i\ell}$ are iid as $e$.
The independence is imposed as in \ref{cond:rv.2},
with requirement \ref{cond:rv.3} on moments.
Write $f_1$, $f_2$, and $f_3$ 
as the respective density functions of 
$T_{i\ell}$, $(T_{i\ell}, \widetilde{X}_i(T_{i\ell}))$,
and $(T_{i\ell_1}, T_{i\ell_2}, \widetilde{X}_i(T_{i\ell_1}), \widetilde{X}_i(T_{i\ell_2}))$.
These density functions are expected be somehow smooth,
as demanded by \ref{cond:rv.4}--\ref{cond:rv.6}.
Without the continuity assumed in \ref{cond:auto.cov}
it would be logically impossible to recover functions $\mu_X$, $v_A$, and $v_C$ by LLS.
Hyper-parameters of LLS are restricted by conditions \ref{cond:kernel.1}--\ref{cond:bandwidth.4}:
the first three exclude certain commonly used kernels (e.g., the Gaussian kernel) 
but admit at least the symmetric Beta family
\citep[][Eq.~2.5]{FanGijbels1996};
the remaining four of \ref{cond:kernel.1}--\ref{cond:bandwidth.4} comprise the cornerstone of the consistency of LLS recovery,
making sure that bandwidths converge at proper rates (as $n$ diverges).
Condition \ref{cond:tau.p.L2} (resp. \ref{cond:tau.p.sup})
implies the convergence rate of PLEASS coefficient estimator
in the $L^2$ (resp. $L^{\infty}$) sense.
Importantly conditions \ref{cond:tau.p.L2} and \ref{cond:tau.p.sup} restrict the divergence rate of $p$ ($=p(n)$)
to be at most $O(n^{1/2}h_A^2)$ if $\|v_A\|_2<1$
and even slower once $\|v_A\|_2\geq 1$.
This restriction on $p$ is pretty close to the setting of \citet[][Theorem~5.3]{DelaigleHall2012b}
who limited the discussion to cases of $\|v_A\|_2<1$ only
(which is reachable by changing the scale on which $X_i$ is measured).
In detail our assumptions are:
\begin{enumerate}[label=(C\arabic*)]
    \item\label{cond:rv.1}
        $\E(L)<\infty$ and $\Pr(L\geq 2)>0$.
    \item\label{cond:rv.2}
        $X_i,T_{i1},\ldots,T_{iL_i}$ and $e_{i1},\ldots,e_{iL_i}$
        are all independent of $L_i$ in the sense that, given $L_i=\ell$, $X_i,T_{i1},\ldots,T_{i\ell}$ and $e_{i1},\ldots,e_{i\ell}$ are all independent and the conditional laws are those of $X$, $T$, and $e$.
    \item\label{cond:rv.3}
        $\E\{X(T)-\mu_X(T)+ \sigma_e e\}^4<\infty$.
    \item\label{cond:rv.4}
        $(\dd^2/\dd t^2)f_1$ exists and is continuous on $\mathbb{T}$.
        The support of $f_1$ is $\mathbb{T}$.
    \item\label{cond:rv.5}
        $(\dd^2/\dd t^2)f_2$ exists and is uniformly continuous on $\mathbb{T}\times\mathbb{R}$.
    \item\label{cond:rv.6}
        $\{\dd^2/(\dd t_1\dd t_2)\}f_3$,
        $(\dd^2/\dd t_1^2)f_3$
        and $(\dd^2/\dd t_2^2)f_3$
        all exist and are uniformly continuous on
        $\mathbb{T}^2\times\mathbb{R}^2$.
    \item\label{cond:auto.cov}
        $\mu_X$ and $v_C$ are both continuous on $\mathbb{T}$,
        and $v_A$ is continuous on $\mathbb{T}^2$.
        Hence $\|\mu\|_{\infty}$, $\|v_C\|_{\infty}$, and $\|v_A\|_{\infty}$ are all finite.
    \item\label{cond:kernel.1}
        The kernel function $\kappa$ in Appendix \ref{appendix:LLS} is 
        symmetric (w.r.t. the $y$ axis) and nonnegative 
        on $\mathbb{R}$ such that
        $\int_{\mathbb{R}}\kappa(t)\dd t=1$.
    \item\label{cond:kernel.2}
        The kernel function $\kappa$ is compactly supported,
        i.e., ${\rm supp}(\kappa)$ is bounded.
    \item\label{cond:kernel.3}
        The Fourier transform of $\kappa$
        is absolutely integrable,
        i.e., 
        $\int_{\mathbb{R}}|\int_{\mathbb{R}}e^{-ist}\kappa(s)\dd s|\dd t<\infty$.
        An implication is 
        the continuity of $\kappa$ (almost everywhere) within ${\rm supp}(\kappa)$.
        Holding \ref{cond:kernel.2} too,
        we automatically have two moment conditions on $\kappa$:
        $\int_{\mathbb{R}}t^2\kappa(t)\dd t<\infty$
        and $\int_{\mathbb{R}}\kappa^2(t)\dd t<\infty$.
    \item\label{cond:bandwidth.1}
        $h_{\mu}\to 0$, 
        $nh_{\mu}^4\to\infty$, and
        $nh_{\mu}^6=O(1)$,
        as $n\to\infty$. 
        Hence $\zeta_{\mu}=n^{-1/2}h_{\mu}^{-1}=o(1)$.
    \item\label{cond:bandwidth.2}
        $h_A\to 0$, 
        $nh_A^6\to\infty$, and
        $nh_A^8=O(1)$,
        as $n\to\infty$.
        Hence $\zeta_A=n^{-1/2}h_A^{-2}=o(1)$.
    \item\label{cond:bandwidth.3}
        $h_{\sigma}\to 0$, 
        $nh_{\sigma}^4\to\infty$, and
        $nh_{\sigma}^6=O(1)$,
        as $n\to\infty$.
        Hence $\zeta_{\sigma}=n^{-1/2}(h_A^{-2}+h_{\sigma}^{-1})=o(1)$.
    \item\label{cond:bandwidth.4}
        $h_C\to 0$, 
        $nh_C^4\to\infty$, and
        $nh_C^6=O(1)$,
        as $n\to\infty$.
        Hence $\zeta_C=n^{-1/2}(h_{\mu}^{-1}+h_C^{-1})=o(1)$.
    \item\label{cond:tau.p.L2}
        As $n\to\infty$,
        $p=p(n)=O(\zeta_A^{-1})$.
        Additional requirements on $p$ vary with the magnitude of $\|v_A\|_2$; 
        they also depend on $\tau_p,$
        the smallest eigenvalue of $\mathbf{D}_p$ which is defined at \eqref{eq:D.p}.
        \begin{itemize}
            \item
		        $O\{\tau_p^{-1}p\|v_A\|_2^{2p}\zeta_C\max(1, \tau_p^{-1}p\|v_A\|_2^{2p})\}$
		        and 
		        $O\{\tau_p^{-1}p^2\|v_A\|_2^{2p}\zeta_A\max(1, \tau_p^{-1}p\|v_A\|_2^{2p})\}$
		        are both of order $o(1)$ I
		        if $\|v_A\|_2\geq 1$;
	        \item if $\|v_A\|_2<1$,
		        then
		        $\tau_p^{-2}\max(\zeta_A,\zeta_C)$
		        and
		        $\tau_p^{-1}\max(\zeta_A,\zeta_C)$
		        are both of order $o(1)$.
        \end{itemize}
    \item\label{cond:tau.p.sup}
        Condition \ref{cond:tau.p.L2} holds with the $L^2$-norm $\|\cdot\|_2$ replaced by the infinity norm $\|\cdot\|_{\infty}$.
\end{enumerate}

The first fourteen of the conditions above are inherited from \citet{YaoMullerWang2005a, YaoMullerWang2005b}.
So is Lemma \ref{lemma:converge.1}
which states the convergence rate of LLS estimators.
We then extend \eqref{eq:converge.vc} to a more general version 
(see Lemma \ref{lemma:converge.cxxb.j}).

\begin{Lemma}[\citealp{YaoMullerWang2005a}, Theorem 1 and Corollary 1; 
    \citealp{YaoMullerWang2005b}, Lemma A.1]\label{lemma:converge.1}
    Under assumptions \ref{cond:rv.1}--\ref{cond:bandwidth.4},
    as $n\to\infty$,
    \begin{align}
        \|\hat{\mu}_X-\mu_X\|_{\infty}
        &=O_p(\zeta_{\mu})=o_p(1),
        \notag\\
        \|\hat{v}_A-v_A\|_{\infty}
        &=O_p(\zeta_A)=o_p(1),
        \notag\\
        |\hat{\sigma}_e^2-\sigma_e^2|
        &=O_p(\zeta_{\sigma})=o_p(1),
        \notag\\
        \intertext{and}
        \|\hat{v}_C-v_C\|_{\infty}
        &=O_p(\zeta_C)=o_p(1),\label{eq:converge.vc}
    \end{align}
    where $\zeta_{\mu}$, $\zeta_A$, $\zeta_{\sigma}$ and $\zeta_C$
    are respectively defined as in 
    \ref{cond:bandwidth.1}--\ref{cond:bandwidth.4}.
\end{Lemma}

\begin{Lemma}\label{lemma:converge.cxxb.j}
    Assume \ref{cond:auto.cov}--\ref{cond:bandwidth.4}
    and that there is a $C>0$ such that for all $n$  we have $p\in [1,C\zeta_A^{-1}]$.
    Then, for each $\epsilon>0$, there are  positive constants $C_1$ and $C_2$ and an integer $n_0>0$ such that,
    for each $n> n_0$,
    \begin{align*}
        \Pr\left[
            \bigcap_{j=1}^{p}\{
                \|\mathcal{V}_A^j(\beta)-\widehat{\mathcal{V}}_A^j(\beta)\|_2
                \leq C_1\|v_A\|_2^{j-1}\zeta_C
                    + C_2(j-1)\|v_A\|_2^{j-1}\zeta_A
            \}
        \right]&\geq 1-\epsilon,
        \\
        \intertext{and}
        \Pr\left[
            \bigcap_{j=1}^{p}\{
                \|\mathcal{V}_A^j(\beta)-\widehat{\mathcal{V}}_A^j(\beta)\|_{\infty}
                \leq C_1\|v_A\|_{\infty}^{j-1}\zeta_C
                    + C_2(j-1)\|v_A\|_{\infty}^{j-1}\zeta_A
            \}
        \right]&\geq 1-\epsilon.
    \end{align*}
\end{Lemma}

\begin{proof}[Proof of Lemma \ref{lemma:converge.cxxb.j}]
    Recall the definitions of $\mathcal{V}_A$ in \eqref{eq:V.A} 
    and of $\widehat{\mathcal{V}}_A$ in \eqref{eq:V.hat.A}. 
    Since $\mathcal{V}_A(\beta)=v_C$ 
    and $\widehat{\mathcal{V}}_A(\beta)=\hat{v}_C$,
    Lemma \ref{lemma:converge.cxxb.j} reduces to \eqref{eq:converge.vc}
    when $j=1$.
    For integer $j\geq 2$ and each $t\in\mathbb{T}$,
    the identity
    \begin{align*}
        |\widehat{\mathcal{V}}_A^j(\beta)(t)&-\mathcal{V}_A^j(\beta)(t)|
        \\
        =&\ |
            \widehat{\mathcal{V}}_A\{\widehat{\mathcal{V}}_A^{j-1}(\beta)-\mathcal{V}_A^{j-1}(\beta)\}(t)
            +(\widehat{\mathcal{V}}_A-\mathcal{V}_A)\{V_A^{j-1}(\beta)\}(t)
        |
        \\
        \leq&\ \|\widehat{\mathcal{V}}_A^{j-1}(\beta)-\mathcal{V}_A^{j-1}(\beta)\|_2
        \left\{\int\hat{v}_A^2(s,t)\dd s\right\}^{1/2}
        \\
        &+ \|\mathcal{V}_A^{j-1}(\beta)\|_2
        \left[\int\{\hat{v}_A(s,t)-v_A(s,t)\}^2\dd s\right]^{1/2}
        \quad\text{(Cauchy-Schwarz)}
    \end{align*}
    implies that 
    \begin{align*}
        \|\mathcal{V}_A^j(\beta)-\widehat{\mathcal{V}}_A^j(\beta)\|_2
        \leq&\ 
        \|\hat{v}_A\|_2
            \|\mathcal{V}_A^{j-1}(\beta)-\widehat{\mathcal{V}}_A^{j-1}(\beta)\|_2
        +
        \|\mathcal{V}_A^{j-1}(\beta)\|_2\|v_A-\hat{v}_A\|_2,
        \\
        \|\mathcal{V}_A^j(\beta)-\widehat{\mathcal{V}}_A^j(\beta)\|_{\infty}
        \leq&\ 
        \|\hat{v}_A\|_{\infty}
            \|\mathcal{V}_A^{j-1}(\beta)-\widehat{\mathcal{V}}_A^{j-1}(\beta)\|_{\infty}
        +
        \|\mathcal{V}_A^{j-1}(\beta)\|_2\|v_A-\hat{v}_A\|_{\infty}.
    \end{align*}
    On iteration these two inequalities give that, respectively,
    \begin{align}
        \|\mathcal{V}_A^j(\beta)&-\widehat{\mathcal{V}}_A^j(\beta)\|_2
        \notag\\
        \leq&\ 
        \|\hat{v}_A\|_2^{j-1}\|
            \mathcal{V}_A(\beta)-\widehat{\mathcal{V}}_A(\beta)\|_2
        +
        \|v_A-\hat{v}_A\|_2\sum_{k=1}^{j-1}
            \|\mathcal{V}_A^k(\beta)\|_2\|\hat{v}_A\|_2^{j-k-1},
        \label{eq:bound.L2.est.Vj}
        \\
        \|\mathcal{V}_A^j(\beta)&-\widehat{\mathcal{V}}_A^j(\beta)\|_{\infty}
        \notag\\
        \leq&\ 
        \|\hat{v}_A\|_{\infty}^{j-1}\|
            \mathcal{V}_A(\beta)-\widehat{\mathcal{V}}_A(\beta)\|_{\infty}
        +
        \|v_A-\hat{v}_A\|_{\infty}\sum_{k=1}^{j-1}
            \|\mathcal{V}_A^k(\beta)\|_2\|\hat{v}_A\|_{\infty}^{j-k-1}.
        \label{eq:bound.sup.est.Vj}
    \end{align}
    For each $\epsilon>0$,
    there is $n_0>0$ such that,
    for all $n>n_0$,
    we have
    \begin{align*}
        1-\epsilon/2
        \leq&\ \Pr(\|\hat{v}_A-v_A\|_2\leq C_0\zeta_A)
        \leq\Pr(\|\hat{v}_A\|_2\leq \|v_A\|_2+C_0\zeta_A),
        \\
        1-\epsilon/2
        \leq&\ \Pr(\|\hat{v}_A-v_A\|_{\infty}\leq C_0\zeta_A)
        \leq\Pr(\|\hat{v}_A\|_{\infty}\leq \|v_A\|_{\infty}+C_0\zeta_A),
        \\
        1-\epsilon/2
        \leq&\ \Pr(\|\hat{v}_C-v_C\|_2\leq C_0\zeta_C),
        \intertext{and}
        1-\epsilon/2
        \leq&\ \Pr(\|\hat{v}_C-v_C\|_{\infty}\leq C_0\zeta_C),
    \end{align*}
    with constant $C_0>0$, by Lemma \ref{lemma:converge.1}.
    It follows from \eqref{eq:bound.L2.est.Vj} that
    \begin{align*}
        1- \epsilon
        \leq\Pr\Bigg[\bigcap_{j=1}^{p}\Bigg\{
            \|\mathcal{V}_A^j(\beta)-\widehat{\mathcal{V}}_A^j(\beta)&\|_2
            \leq(\|v_A\|_2+C_0\zeta_A)^{j-1}C_0\zeta_C
        \\
            &+
            C_0\zeta_A\sum_{k=1}^{j-1}
                \|v_A\|_2^k\|\beta\|_2(\|v_A\|_2+C_0\zeta_A)^{j-k-1}
        \Bigg\}\Bigg]
        \\
        \leq\Pr\Bigg[\bigcap_{j=1}^{p}\Bigg\{
            \|\mathcal{V}_A^j(\beta)-\widehat{\mathcal{V}}_A^j(\beta)&\|_2
            \leq C_0(1+C_0\zeta_A/\|v_A\|_2)^{j-1}\|v_A\|_2^{j-1}\zeta_C
        \\
            &+
            C_0\|\beta\|_2\zeta_A\|v_A\|_2^{j-1}\sum_{k=1}^{j-1}
                (1+C_0\zeta_A/\|v_A\|_2)^{j-k-1}
        \Bigg\}\Bigg]
        \\
        \leq\Pr\Bigg[\bigcap_{j=1}^{p}\{
            \|\mathcal{V}_A^j(\beta)-\widehat{\mathcal{V}}_A^j(\beta)&\|_2
            \leq C_1\|v_A\|_2^{j-1}\zeta_C
        \\
            &+
            C_2(j-1)\|v_A\|_2^{j-1}\zeta_A
        \}\Bigg],
        \quad\text{(if $p\leq C\zeta_A^{-1}$ with fixed $C>0$)}
    \end{align*}
    where $C_1=C_0\exp(CC_0/\|v_A\|_2)\geq C_0\exp(CC_0/\|v_A\|_{\infty})$ 
    and $C_2=\|\beta\|_2C_1$.
    It is worth noting that we have assumed that the range of $p$ is constrained in $[1, C\zeta_A^{-1}]$; the quantity
    $(1+C_0\zeta_A/\|v_A\|_2)^{p}$ may not be bounded if $p$ diverges too fast.
    Similarly, inequality
    \eqref{eq:bound.sup.est.Vj} implies that,
    for $1\leq p\leq C\zeta_A^{-1}$,
    $$
        \Pr\left[\bigcap_{j=1}^{p}\{
            \|\mathcal{V}_A^j(\beta)-\widehat{\mathcal{V}}_A^j(\beta)\|_{\infty}
            \leq C_1\|v_A\|_{\infty}^{j-1}\zeta_C
            +
            C_2(j-1)\|v_A\|_{\infty}^{j-1}\zeta_A
        \}\right]
        \geq 1-\epsilon.
    $$
\end{proof}

\begin{proof}[Proof of Theorem \ref{thm:converge.beta}]
    The following alternative expression for $\beta_p$ \eqref{eq:beta.p}, drawn from \citet[][Eq.~3.6]{DelaigleHall2012b}, 
    dramatically facilitates our further moves:
    \begin{equation}\label{eq:beta.p.alter}
    	\beta_p =\beta_p(\cdot)
    	= [\mathcal{V}_A(\beta)(\cdot),\ldots,\mathcal{V}_A^p(\beta)(\cdot)]\DDD_p^{-1}\aalpha_p,
    \end{equation}
    where
    \begin{align}
        \DDD_p=[d_{j_1,j_2}]_{1\leq j_1,j_2\leq p},
        \label{eq:D.p}\\
        \aalpha_p=[\alpha_1,\ldots,\alpha_p]^\top,
        \label{eq:alpha.p}
    \end{align}
    with
    $   
        d_{j_1,j_2}
        =\int\mathcal{V}_A^{j_1+1}(\beta)\mathcal{V}_A^{j_2}(\beta)
        =\int\mathcal{V}_A^{j_1}(\beta)\mathcal{V}_A^{j_2+1}(\beta)
    $
    and
    $
        \alpha_j
        =\int\mathcal{V}_A(\beta)\mathcal{V}_A^j(\beta)
        =\int v_C\mathcal{V}_A^j(\beta)
    $.
    As is known,
    $\DDD_p^{-1}$ and $\aalpha_p$ are bounded, respectively, as
    \begin{equation}\label{eq:bound.D.p.inv}
        \|\DDD_p^{-1}\|_2=\tau_p^{-1}
    \end{equation}
    and
    \begin{align}
        \|\aalpha_p\|_2
        =\left[\sum_{j=1}^p\left\{\int v_C\mathcal{V}_A^j(\beta)\right\}^2\right]^{1/2}
        \leq&\ \left[\sum_{j=1}^p\|v_C\|_2^2\|\mathcal{V}_A^j(\beta)\|_2^2\right]^{1/2}
        \quad\text{(Cauchy-Schwarz)}
        \notag\\
        =&\ \begin{cases}
            O(p^{1/2}\|v_A\|_2^p)
            &\text{if }\|v_A\|_2\geq 1
            \\
            O(1)
            &\text{if }\|v_A\|_2< 1.
        \end{cases}
        \label{eq:bound.alpha.p}
    \end{align}
    Corresponding to \eqref{eq:beta.p.alter}, 
    $\hat{\beta}_p$ at \eqref{eq:beta.p.hat} is rewritten as
    \begin{equation}\label{eq:beta.p.hat.alter}
        \hat{\beta}_p=\hat{\beta}_p(\cdot)
        = [\widehat{\mathcal{V}}_A(\beta)(\cdot),\ldots,\widehat{\mathcal{V}}_A^p(\beta)(\cdot)]
        \widehat{\DDD}_p^{-1}\hat{\aalpha}_p,
    \end{equation}
    in which
    $\widehat{\DDD}_p=[\hat{d}_{j_1,j_2}]_{1\leq j_1, j_2 \leq p}$
    and
    $\hat{\aalpha}_p=[\hat{\alpha}_1,\ldots,\hat{\alpha}_p]^\top$
    are respective empirical counterparts of $\DDD_p$ at \eqref{eq:D.p}
    and $\aalpha_p$ at \eqref{eq:alpha.p},
    with 
    $
        \hat{d}_{j_1,j_2}
        =\int \widehat{\mathcal{V}}_A^{j_1+1}(\beta)\widehat{\mathcal{V}}_A^{j_2}(\beta)
    $
    and
    $
        \hat{\alpha}_j
        =\int \widehat{\mathcal{V}}_A(\beta)\widehat{\mathcal{V}}_A^j(\beta)
        =\int \hat{v}_C\widehat{\mathcal{V}}_A^j(\beta)
    $.
    
    Observe that,
    by the Cauchy-Schwarz inequality,
    \begin{align*}
        |\alpha_j-\hat{\alpha}_j|
        =&\ 
        \left|
            \int (v_C-\hat{v}_C)\mathcal{V}_A^j(\beta)
        \right|
        +
        \left|
            \int \hat{v}_C\{\widehat{\mathcal{V}}_A^j(\beta)-\mathcal{V}_A^j(\beta)\}
        \right|
        \\
        \leq&\ 
        \|\beta\|_2\|v_A\|_2^j\|\hat{v}_C-v_C\|_2
        +
        \|\hat{v}_C\|_2\|\widehat{\mathcal{V}}_A^j(\beta)-\mathcal{V}_A^j(\beta)\|_2. 
    \end{align*}
    For every $\epsilon>0$ and $1\leq p\leq C\zeta_A^{-1}$,
    there is $n_0>0$ such that,
    $\forall n>n_0$,
    $$
        1-\epsilon
        \leq\Pr\left[\bigcap_{j=1}^p\{
            |\alpha_j-\hat{\alpha}_j|\leq 
                C_3\|v_A\|_2^{j-1}\zeta_C+C_4(j-1)\|v_A\|_2^{j-1}\zeta_A
        \}\right],
        \quad\text{(by Lemmas \ref{lemma:converge.1} and \ref{lemma:converge.cxxb.j})}
    $$
    with constants $C_3,C_4>0$.
    Analogously,
    writing $\Delta_{jk}=\hat{d}_{jk}-d_{jk}$,
    the Cauchy-Schwarz inequality implies that
    \begin{align*}
        |\Delta_{jk}|
        \leq&\ \|\widehat{\mathcal{V}}_A^{j+1}(\beta)-\mathcal{V}_A^{j+1}(\beta)\|_2
        \|\widehat{\mathcal{V}}_A^k(\beta)\|_2
        +
        \|\widehat{\mathcal{V}}_A^k(\beta)-\mathcal{V}_A^k(\beta)\|_2
        \|\mathcal{V}_A^{j+1}(\beta)\|_2
        \\
        \leq&\ 
        \|\widehat{\mathcal{V}}_A^{j+1}(\beta)-\mathcal{V}_A^{j+1}(\beta)\|_2
        \|\hat{v}_A\|_2^k\|\beta\|_2
        +
        \|\widehat{\mathcal{V}}_A^k(\beta)-\mathcal{V}_A^k(\beta)\|_2
        \|v_A\|_2^{j+1}\|\beta\|_2,
    \end{align*}
    and further, by Lemmas \ref{lemma:converge.1} and \ref{lemma:converge.cxxb.j},
    as long as $1\leq p\leq C\zeta_A^{-1}$,
    \begin{align*}
        1-\epsilon
        \leq\Pr\Bigg[\bigcap_{j,k=1}^p\{
            |\Delta_{jk}|
            \leq&\ 
            \|\widehat{\mathcal{V}}_A^{j+1}(\beta)-\mathcal{V}_A^{j+1}(\beta)\|_2
            (\|v_A\|_2+C_0\zeta_A^{-1})^k\|\beta\|_2
        \\
            &+
            \|\widehat{\mathcal{V}}_A^k(\beta)-\mathcal{V}_A^k(\beta)\|_2
            \|v_A\|_2^{j+1}\|\beta\|_2
        \}\Bigg]
        \\
        \leq\Pr\Bigg[\bigcap_{j,k=1}^p\{
            |\Delta_{jk}|
            \leq&\ 
            C_5\|v_A\|_2^{j+k}\zeta_C
            +C_6\max(j,k-1)\|v_A\|_2^{j+k}\zeta_A
        \}\Bigg],
    \end{align*}
    where $C_5$ and $C_6$ are positive constants.  
    Thus, 
    if $\DDelta_p=[\Delta_{jk}]_{p\times p}=\widehat{\mathbf{D}}_p-\mathbf{D}_p$,
    then
    \begin{align}
        \|\DDelta_p\|_2^2
        \leq&\ \sum_{1\leq j,k\leq p}\Delta_{jk}^2
        \notag\\
        =&\ O_p\left(
                \zeta_C^2\sum_{1\leq j,k\leq p}\|v_A\|_2^{2j+2k}
            \right)
        + O_p\left[
                \zeta_A^2\sum_{1\leq j,k\leq p}\max\{j^2,(k-1)^2\}\|v_A\|_2^{2j+2k}
            \right]
        \notag\\
        =&\ \begin{cases}
            O_p(p^2\|v_A\|_2^{4p}\zeta_C^2) + O_p(p^4\|v_A\|_2^{4p}\zeta_A^2)
            &\text{if }\|v_A\|_2\geq 1
            \\
            O_p(\zeta_C^2) + O_p(\zeta_A^2)
            &\text{if }\|v_A\|_2< 1.
        \end{cases}
        \label{eq:sum.delta.jk}
    \end{align}
    In a similar manner,
    one proves that
    \begin{align}
        \|\hat{\aalpha}_p-\aalpha_p\|_2^2
        =&\ \sum_{1\leq j\leq p}|\hat{\alpha}_j-\alpha_j|^2
        \notag\\
        =&\ O_p\left(
            C_1\zeta_C^2\sum_{1\leq j\leq p}\|v_A\|_2^{2j-2}
        \right)
        + O_p\left\{
            \zeta_A^2\sum_{1\leq j\leq p}(j-1)^2\|v_A\|_2^{2j-2}
        \right\}
        \notag\\
        =&\ \begin{cases}
            O_p(p\|v_A\|_2^{2p}\zeta_C^2) + O_p(p^3\|v_A\|_2^{2p}\zeta_A^2)
            &\text{if }\|v_A\|_2\geq 1
            \\
            O_p(\zeta_C^2) + O_p(\zeta_A^2)
            &\text{if }\|v_A\|_2< 1.
        \end{cases}
        \label{eq:diff.alpha.j}
    \end{align}
    Denote by $\tau_p$ the smallest eigenvalue of $\mathbf{D}_p$.
    Notice that,
    for $p=p(n)=O(\zeta_A^{-1})$,
    \begin{align*}
        \|\mathbf{D}_p^{-1}\DDelta_p\|_2
        \leq&\ \tau_p^{-1}\|\DDelta_p\|_2
        \notag\\
        =&\ \begin{cases}
            O_p(\tau_p^{-1}p\|v_A\|_2^{2p}\zeta_C) + O_p(\tau_p^{-1}p^2\|v_A\|_2^{2p}\zeta_A)
            &\text{if }\|v_A\|_2\geq 1
            \\
            O_p(\tau_p^{-1}\zeta_C) + O_p(\tau_p^{-1}\zeta_A)
            &\text{if }\|v_A\|_2< 1.
        \end{cases}
        \quad\text{(by \eqref{eq:bound.D.p.inv} and \eqref{eq:sum.delta.jk})}
    \end{align*}
    Provided that \ref{cond:tau.p.L2} holds,
    for sufficiently large $n$,
    one has $\tau_p^{-1}\|\DDelta_p\|_2<\gamma$,
    for some $\gamma\in (0,1)$.
    In this case,
    \citet[][Eq.~7.18]{DelaigleHall2012b} argued that,
    as $n$ goes to infinity,
    $$
        \widehat{\mathbf{D}}_p^{-1}
        =\{\mathbf{I}-\mathbf{D}_p^{-1}\DDelta_p+O_p(\tau_p^{-2}\|\DDelta_p\|_2^2)\}\mathbf{D}_p^{-1},
    $$
    which can be rewritten as
    \begin{align}
        \|\widehat{\mathbf{D}}_p^{-1}&-\mathbf{D}_p^{-1}\|_2
        \notag\\
        =&\ \|\{O_p(\tau_p^{-2}\|\DDelta_p\|_2^2) - \mathbf{D}_p^{-1}\DDelta_p\}\mathbf{D}_p^{-1}\|_2 
        \notag\\
        =&\ \begin{cases}
            \tau_p^{-1}\|
                O_p(\tau_p^{-2}p^2\|v_A\|_2^{4p}\zeta_C^2)
                +O_p(\tau_p^{-2}p^4\|v_A\|_2^{4p}\zeta_A^2) 
                -\mathbf{D}_p^{-1}\DDelta_p
            \|_2
            &\text{if }\|v_A\|_2\geq 1
            \\
            \tau_p^{-1}\|
                O_p(\tau_p^{-2}\zeta_C^2)
                +O_p(\tau_p^{-2}\zeta_A^2) 
                -\mathbf{D}_p^{-1}\DDelta_p
            \|_2
            &\text{if }\|v_A\|_2< 1
        \end{cases}
        \quad\text{(by \eqref{eq:sum.delta.jk})}
        \notag\\
        =&\ \begin{cases}
            O_p(\tau_p^{-2}p\|v_A\|_2^{2p}\zeta_C) + O_p(\tau_p^{-2}p^2\|v_A\|_2^{2p}\zeta_A)
            &\text{if }\|v_A\|_2\geq 1
            \\
            O_p(\tau_p^{-2}\zeta_C) + O_p(\tau_p^{-2}\zeta_A)
            &\text{if }\|v_A\|_2< 1.
        \end{cases}
        \quad\text{(by \ref{cond:tau.p.L2})}
        \label{eq:diff.D.inverse}
    \end{align}
    Combining \eqref{eq:bound.D.p.inv}, \eqref{eq:bound.alpha.p},
    \eqref{eq:diff.alpha.j} and \eqref{eq:diff.D.inverse},
    one  obtains
    \begin{align}
        \|\widehat{\mathbf{D}}_p^{-1}\hat{\aalpha}_p&-\mathbf{D}_p^{-1}\aalpha_p\|_2
        \notag\\
        \leq&\ \|\widehat{\mathbf{D}}_p^{-1}-\mathbf{D}_p^{-1}\|_2\|\aalpha_p\|_2
            +\|\widehat{\mathbf{D}}_p^{-1}\|_2\|\hat{\aalpha}_p-\aalpha_p\|_2
        \notag\\
        =&\ \begin{cases}
            O_p(\tau_p^{-2}p^{3/2}\|v_A\|_2^{3p}\zeta_C)
            + O_p(\tau_p^{-2}p^{5/2}\|v_A\|_2^{3p}\zeta_A)
            \\
            \qquad + O_p(\tau_p^{-1}p^{1/2}\|v_A\|_2^{p}\zeta_C)
            + O_p(\tau_p^{-1}p^{3/2}\|v_A\|_2^{p}\zeta_A)
            &\text{if }\|v_A\|_2\geq 1
            \\
            O_p(\tau_p^{-2}\zeta_C) 
            + O_p(\tau_p^{-2}\zeta_A)
            + O_p(\tau_p^{-1}\zeta_C)
            + O_p(\tau_p^{-1}\zeta_A)
            &\text{if }\|v_A\|_2< 1.
        \end{cases}
        \label{eq:diff.dalpha}
    \end{align}
    
    Next, for each $t\in\mathbb{T}$, we have
    \begin{align*}
        |\hat{\beta}_p(t)-\beta_p(t)|^2
        =\Bigg|
                [\widehat{\mathcal{V}}_A(\beta)(t),&\ldots,\widehat{\mathcal{V}}_A^p(\beta)(t)]
                \widehat{\DDD}_p^{-1}\hat{\aalpha}_p
        \\
                &-
                [\mathcal{V}_A(\beta)(t),\ldots,\mathcal{V}_A^p(\beta)(t)]
                \DDD_p^{-1}\aalpha_p
            \Bigg|^2
        \\
        \leq\Bigg|
            \|\widehat{\DDD}_p^{-1}\hat{\aalpha}_p&-\DDD_p^{-1}\aalpha_p\|_2
            \left[\sum_{j=1}^p\{\widehat{\mathcal{V}}_A^j(\beta)(t)\}^2\right]^{1/2}
        \\
            &+
            \|\DDD_p^{-1}\aalpha_p\|_2
            \left[\sum_{j=1}^p
                \{\widehat{\mathcal{V}}_A^j(\beta)(t)-\mathcal{V}_A^j(\beta)(t)\}^2\right]^{1/2}
        \Bigg|^2
        \\
        \leq 
            2\|\widehat{\mathbf{D}}_p^{-1}\hat{\aalpha}_p&-\mathbf{D}_p^{-1}\aalpha_p\|_2^2
            \left[\sum_{j=1}^p\{\widehat{\mathcal{V}}_A^j(\beta)(t)\}^2\right]
        \\
            &+
            2\|\mathbf{D}_p^{-1}\aalpha_p\|_2^2
            \left[\sum_{j=1}^p
                \{\widehat{\mathcal{V}}_A^j(\beta)(t)-\mathcal{V}_A^j(\beta)(t)\}^2\right].
    \end{align*}
    Thus $\|\hat{\beta}_p-\beta_p\|_2$ is bounded as below:
    \begin{align}
        \|\hat{\beta}_p-\beta_p\|_2^2
        \leq&\ 2\|\widehat{\mathbf{D}}_p^{-1}\hat{\aalpha}_p-\mathbf{D}_p^{-1}\aalpha_p\|_2^2
            \sum_{j=1}^p\|\mathcal{V}_A^j(\beta)\|_2^2
            + 
            2\|\mathbf{D}_p^{-1}\aalpha_p\|_2^2
            \sum_{j=1}^p\|\mathcal{V}_A^j(\beta)-\widehat{\mathcal{V}}_A^j(\beta)\|_2^2
        \notag\\
        \leq&\ 2\|\widehat{\mathbf{D}}_p^{-1}\hat{\aalpha}_p-\mathbf{D}_p^{-1}\aalpha_p\|_2^2
            \sum_{j=1}^p\|\mathcal{V}_A^j(\beta)\|_2^2
        \label{eq:beta.dist.1}\\
        &+ 2\tau_p^{-2}\|\aalpha_p\|_2^2
            \sum_{j=1}^p \|\widehat{\mathcal{V}}_A^j(\beta)-\mathcal{V}_A^j(\beta)\|_2^2.
        \label{eq:beta.dist.2}
    \end{align}
    Owing to \eqref{eq:diff.dalpha},
    $$
        \eqref{eq:beta.dist.1}
        =\begin{cases}
            O_p(\tau_p^{-4}p^{4}\|v_A\|_2^{8p}\zeta_C^2)
            + O_p(\tau_p^{-4}p^{6}\|v_A\|_2^{8p}\zeta_A^2)
            \\
            \qquad + O_p(\tau_p^{-2}p^{2}\|v_A\|_2^{4p}\zeta_C^2)
            + O_p(\tau_p^{-2}p^{4}\|v_A\|_2^{4p}\zeta_A^2)
            &\text{if }\|v_A\|_2\geq 1
            \\
            O_p(\tau_p^{-4}\zeta_C^2) 
            + O_p(\tau_p^{-4}\zeta_A^2)
            + O_p(\tau_p^{-2}\zeta_C^2)
            + O_p(\tau_p^{-2}\zeta_A^2)
            &\text{if }\|v_A\|_2< 1;
        \end{cases}
    $$
    the rate of \eqref{eq:beta.dist.2}
    is given by 
    \eqref{eq:bound.alpha.p} and \ref{lemma:converge.cxxb.j} jointly,
    i.e.,
    $$
        \eqref{eq:beta.dist.2}
        =\begin{cases}
            O_p(\tau_p^{-2}p^2\|v_A\|_2^{4p}\zeta_C^2)
            + O_p(\tau_p^{-2}p^4\|v_A\|_2^{4p}\zeta_A^2)
            &\text{if }\|v_A\|_2\geq 1
            \\
            O_p(\tau_p^{-2}\zeta_C^2)
            + O_p(\tau_p^{-2}\zeta_A^2)
            &\text{if }\|v_A\|_2< 1.
        \end{cases}
    $$
    In this way we deduce
    \begin{align}
       \|\hat{\beta}_p-\beta_p\|_2^2
       =\begin{cases}
            O_p(\tau_p^{-4}p^{4}\|v_A\|_2^{8p}\zeta_C^2)
            + O_p(\tau_p^{-4}p^{6}\|v_A\|_2^{8p}\zeta_A^2)
            \\
            \qquad + O_p(\tau_p^{-2}p^{2}\|v_A\|_2^{4p}\zeta_C^2)
            + O_p(\tau_p^{-2}p^{4}\|v_A\|_2^{4p}\zeta_A^2)
            &\text{if }\|v_A\|_2\geq 1
            \\
            O_p(\tau_p^{-4}\zeta_C^2) 
            + O_p(\tau_p^{-4}\zeta_A^2)
            + O_p(\tau_p^{-2}\zeta_C^2)
            + O_p(\tau_p^{-2}\zeta_A^2)
            &\text{if }\|v_A\|_2< 1.
        \end{cases}
        \label{eq:beta.dist}
    \end{align}
    Condition \ref{cond:tau.p.L2} then implies that both 
    \eqref{eq:beta.dist.1} and \eqref{eq:beta.dist.2} 
    converge to 0 in probability.
    The consistency of PLEASS estimators in the $L^2$ sense follows, 
    from the $L^2$ convergence of $\beta_p$ to $\beta$ 
    (\citealp[][Theorem~3.2]{DelaigleHall2012b}).
    
   Finally,
    we bound the estimation error in the supremum metric:
    \begin{align}
        \|\hat{\beta}_p&-\beta_p\|_{\infty}^2
        \notag\\
        =&\ \Bigg\|    
                [\widehat{\mathcal{V}}_A(\beta),\ldots,\widehat{\mathcal{V}}_A^p(\beta)]
                (\widehat{\DDD}_p^{-1}\hat{\aalpha}_p-\DDD_p^{-1}\aalpha_p)
                +
                [\widehat{\mathcal{V}}_A(\beta)-\mathcal{V}_A(\beta),\ldots,
                \widehat{\mathcal{V}}_A^p(\beta)-\mathcal{V}_A^p(\beta)]
                \DDD_p^{-1}\aalpha_p
        \Bigg\|_{\infty}
        \notag\\
        \leq&\ \left[
            \|\widehat{\DDD}_p^{-1}\hat{\aalpha}_p-\DDD_p^{-1}\aalpha_p\|_2
            \left\{\sum_{j=1}^p\|\widehat{\mathcal{V}}_A^j(\beta)\|_{\infty}^2\right\}^{1/2}
            +
            \|\DDD_p^{-1}\aalpha_p\|_2
            \left\{\sum_{j=1}^p
                \|\widehat{\mathcal{V}}_A^j(\beta)-\mathcal{V}_A^j(\beta)\|_{\infty}^2\right\}^{1/2}
        \right]^2
        \notag\\
        \leq&\ 2\|\widehat{\mathbf{D}}_p^{-1}\hat{\aalpha}_p-\mathbf{D}_p^{-1}\aalpha_p\|_2^2
            \sum_{j=1}^p\|\mathcal{V}_A^j(\beta)\|_{\infty}^2
            + 
            2\tau_p^{-2}\|\aalpha_p\|_2^2
            \sum_{j=1}^p\|\mathcal{V}_A^j(\beta)-\widehat{\mathcal{V}}_A^j(\beta)\|_{\infty}^2
        \notag\\
        \leq&\ 2\|\widehat{\mathbf{D}}_p^{-1}\hat{\aalpha}_p-\mathbf{D}_p^{-1}\aalpha_p\|_2^2
            \sum_{j=1}^p\|\mathcal{V}_A^j(\beta)\|_{\infty}^2
        \notag
        \quad
        \text{(different from \eqref{eq:beta.dist.1} only in the metric)}
        \\
        &+2\tau_p^{-2}\|\aalpha_p\|_2^2
            \sum_{j=1}^p \|\widehat{\mathcal{V}}_A^j(\beta)-\mathcal{V}_A^j(\beta)\|_{\infty}^2
        \quad
        \text{(different from \eqref{eq:beta.dist.2} only in the metric)}
        \notag\\
        =&\ \begin{cases}
            O_p(\tau_p^{-4}p^{4}\|v_A\|_{\infty}^{8p}\zeta_C^2)
            + O_p(\tau_p^{-4}p^{6}\|v_A\|_{\infty}^{8p}\zeta_A^2)
            \\
            \qquad + O_p(\tau_p^{-2}p^{2}\|v_A\|_{\infty}^{4p}\zeta_C^2)
            + O_p(\tau_p^{-2}p^{4}\|v_A\|_{\infty}^{4p}\zeta_A^2)
            &\text{if }\|v_A\|_{\infty}\geq 1
            \\
            O_p(\tau_p^{-4}\zeta_C^2) 
            + O_p(\tau_p^{-4}\zeta_A^2)
            + O_p(\tau_p^{-2}\zeta_C^2)
            + O_p(\tau_p^{-2}\zeta_A^2)
            &\text{if }\|v_A\|_{\infty}< 1.
        \end{cases}
        \notag
    \end{align}
    That is,
    the upper bound for $\|\hat{\beta}_p-\beta_p\|_{\infty}$
    can be obtained from \eqref{eq:beta.dist}
    by replacing $\|v_A\|_2$ with $\|v_A\|_{\infty}$.
    Condition \ref{cond:tau.p.sup} completes the proof for 
    the zero-convergence of $\|\hat{\beta}_p-\beta\|_{\infty}$,
    as long as we assume $\|\beta_p-\beta\|_{\infty}\to 0$
    as $p\to\infty$.
\end{proof}

\begin{proof}[Proof of Theorem \ref{thm:converge.eta}]
	Recall the definitions of $\beta_p$ at \eqref{eq:beta.p}
	and $\hat{\beta}_p$ at \eqref{eq:beta.p.hat}.
	Introduce $\bm{s}_p = [\mathcal{V}_A(w_1),\ldots, \mathcal{V}_A(w_p)]^\top$
	and its empirical version
	$\hat{\bm{s}}_p = [\widehat{\mathcal{V}}_A(\hat{w}_1),\ldots, \widehat{\mathcal{V}}_A(\hat{w}_p)]^\top$.
	Note the identities that
	$\bm{c}_p^\top\bm{s}_p = \mathcal{V}_A(\beta_p)$
	and 
	$\hat{\bm{c}}_p^\top\hat{\bm{s}}_p
	= \widehat{\mathcal{V}}_A(\hat{\beta}_p)$.
    Thus,
    conditions \ref{cond:rv.1}--\ref{cond:tau.p.L2} jointly ensure that,
    for arbitrarily given $L^*, T_1^*,\ldots,T_{L^*}$,
    \begin{align*}
        \|
    		\mathbf{H}_p^*\bm{c}_p
    		&- \widehat{\mathbf{H}}_p^*\hat{\bm{c}}_p
    	\|_2^2
    	\\
    	\leq&\ L^* \|
    	    \bm{s}_p^\top\bm{c}_p
    		- \hat{\bm{s}}_p^\top\hat{\bm{c}}_p
    	\|_{\infty}^2
    	\\
    	=&\ L^*\sup_{t\in\mathbb{T}}\left|
    	    \int v_A(s,t)\{\beta_p(s)-\hat{\beta}_p(s)\}\dd s
    	    +\int (v_A-\hat{v}_A)(s,t)\hat{\beta}_p(s)\dd s
    	\right|^2
    	\\
    	\leq&\ L^*\sup_{t\in\mathbb{T}}\Bigg|
    	    \left\{\int v^2_A(s,t)\dd s\right\}^{1/2}
    	    \|\beta_p-\hat{\beta}_p\|_2
    	    +\left\{\int (v_A-\hat{v}_A)^2(s,t)\dd s\right\}^{1/2}
    	    \|\hat{\beta}_p\|_2
    	\Bigg|^2
    	\\
    	\leq&\ L^*(
    	    \|v_A\|_{\infty}\|\beta_p-\hat{\beta}_p\|_2
    	    +\|v_A-\hat{v}_A\|_{\infty}\|\hat{\beta}_p\|_2
    	)^2
    	\\
    	\to&_p\ 0.
        \quad\text{(by Lemma \ref{lemma:converge.1} and Theorem \ref{thm:converge.beta})}
    \end{align*}
    The convergence to 0 (in probability and conditional on $L^*$ and $T_1^*,\ldots,T_{L^*}$) 
    of $\hat{\eta}_p(X^*)-\tilde{\eta}_{\infty}(X^*)$
    (with $\hat{\eta}_p(X^*)$ at \eqref{eq:eta.p.hat}
    and $\tilde{\eta}_{\infty}(X^*)$ at \eqref{eq:eta.inf.tilde})
    follows from 
    Lemma \ref{lemma:converge.1}
    and the continuous mapping and Slutsky's theorems.
    Since  $L^*$ and $T_1^*,\ldots,T_{L^*}$ are arbitrary  the dominated convergence theorem
    enables us to drop the conditioning.
    This completes the proof of
    Theorem \ref{thm:converge.eta}.
\end{proof}

\begin{proof}[Proof of Corollary \ref{thm:asymptotic.normal}]
    Recall $\eta_p(X^*)$ at \eqref{eq:eta.p},
    $\tilde{\eta}_p(X^*)$ at \eqref{eq:eta.p.tilde}
    and $\tilde{\eta}_{\infty}(X^*)$ at \eqref{eq:eta.inf.tilde}.
    As discussed in the last paragraph of Section \ref{sec:estimation},
    $[\tilde{\xi}_1^*-\xi_1^*,\ldots,\tilde{\xi}_p^*-\xi_p^*]^\top
    \sim \mathcal{N}(\bm{0},
    	\II_p-\mathbf{H}_p^{*\top}\SSigma_{\widetilde{X}^*}^{-1}\mathbf{H}_p^*)$.
    It follows that 
    $
        \tilde{\eta}_p(X^*)-\eta_p(X^*)
        \sim\mathcal{N}
            \{0,\bm{c}_p^\top(\II_p-
                \mathbf{H}_p^{*\top}\SSigma_{\widetilde{X}^*}^{-1}\mathbf{H}_p^*)
                \bm{c}_p\}
    $
    and further that
    $\hat{\eta}_p(X^*)-\eta_p(X^*)$ converges (in distribution) to $\mathcal{N}(0,\omega)$
    as $n\to\infty$,
    by Theorem \ref{thm:converge.eta}.
    Slutsky's theorem now completes the proof.
\end{proof}

\bibliography{mybibfile}

@string{aos="Ann. Stat."}

@string{cils="Chemometrics Intell. Lab. Syst."}

@string{csda="Comput. Stat. Data Anal."}

@string{ejs="Electron. J. Stat."}

@string{jasa="J. Am. Stat. Assoc."}

@string{jat="J. Approx. Theory"}

@string{jcgs="J. Comput. Graph. Stat."}

@string{jmva="J. Multivariate Anal."}

@string{jrssb="J. R. Stat. Soc. Ser. B-Stat. Methodol."}

@string{jscs="J. Stat. Comput. Simul."}

@article{AguileraEscabiasPredaSaporta2010,
    title = "Using basis expansions for estimating functional {PLS} regression: 
        {A}pplications with chemometric data",
    journal = cils,
    volume = "104",
    number = "2",
    pages = "289--305",
    year = "2010",
    doi = "10.1016/j.chemolab.2010.09.007",
    author = "Ana M. Aguilera and Manuel Escabias and Cristian Preda and Gilbert Saporta",
}

@phdthesis{Albaqshi2017,
    title = {Generalized Partial Least Squares Approach for 
        Nominal Multinomial Logit Regression Models
        with a Functional Covariate},
    author = {Amani Mohammed H Albaqshi},
    school = {University of Northern Colorado},
    year = {2017},
}

@article{Baillo2009,
    author = {Amparo Ba\'{i}llo},
    title = {A note on functional linear regression},
    journal = jscs,
    volume = {79},
    pages = {657--669},
    year  = {2009},
    doi = {10.1080/00949650701836765},
}

@article{CravenWahba1979,
	Author = {Craven, Peter and Wahba, Grace},
	Title = {Smoothing noisy data with spline functions},
	Journal = {Numer. Math.},
	Volume = {31},
	Pages = {377--403},
	DOI = {10.1007/BF01404567},
	Year = {1979} 
}

@article{deJong1993,
    title = "{SIMPLS}: An alternative approach to partial least squares regression",
    journal = cils,
    volume = "18",
    pages = "251--263",
    year = "1993",
    doi = "10.1016/0169-7439(93)85002-X",
    author = "de Jong, Sijmen ",
}

@article{DelaigleHall2012a,
	author = {Delaigle, Aurore and Hall, Peter},
	title = {Achieving near perfect classification for functional data},
	journal = jrssb,
	volume = {74},
	year = {2012},
	pages = {267--286},
	doi = {10.1111/j.1467-9868.2011.01003.x},
}

@article{DelaigleHall2012b,
	author = {Delaigle, Aurore and Hall, Peter},
	title = {Methodology and theory for partial least squares
	    applied to functional data},
	journal = aos,
	volume = {40},
	year = {2012},
	pages = {322--352},
	doi = {10.1214/11-AOS958},
}

@book{FanGijbels1996,
    Author = {Fan, Jianqing and Gijbels, Irene},
    Title = {Local Polynomial Modelling and Its Applications},
    Publisher = {Chapman \& Hall/CRC},
    Address = {Boca Raton},
    Series = {Monographs on Statistics and Applied Probability},
    Year = {1996},
}

@article{GoldsmithBobCrainiceanuCaffoReich2011,
    title={Penalized functional regression},
    author={Jeff Goldsmith and Jennifer Bob and Ciprian M. Crainiceanu and Brian Caffo and Daniel Reich},
    journal=jcgs,
    volume={20},
    pages={830--851},
    year={2011},
    doi={10.1198/jcgs.2010.10007},
}

@article{Goutis1998,
    title={Second‐derivative functional regression with applications to near infra‐red spectroscopy},
    author={Constantinos Goutis},
    journal=jrssb,
    volume={60},
    pages={103--114},
    year={1998},
    doi={10.1111/1467-9868.00111}
}

@article{HallMullerWang2006,
    author = "Hall, Peter and Müller, Hans-Georg and Wang, Jane-Ling",
    doi = "10.1214/009053606000000272",
    journal = aos,
    pages = "1493--1517",
    title = "Properties of principal component methods for functional and longitudinal data analysis",
    volume = "34",
    year = "2006",
}

@article{Harville1976,
	author = "Harville, David",
	doi = "10.1214/aos/1176343414",
	journal = aos,
	pages = "384--395",
	title = "Extension of the {G}auss-{M}arkov Theorem 
		to Include the Estimation of Random Effects",
	volume = "4",
	year = "1976"
}

@incollection{Hochstrasser1972,
    author = {Urs W. Hochstrasser},
    title = {Orthogonal Polynomials},
    editor = {Milton Abramowitz and Irene A. Stegun},
    booktitle = {Handbook of Mathematical Functions with Formulas, Graphs, and Mathematical Tables},
    year = {1972},
    series = {Applied Mathematics Series 55},
    publisher = {Dover Publications, Inc.},
    address = {New York},
    note = {Tenth original printing with corrections},
    pages = {773--802},
}

@article{JamesHastieSugar2000,
	title = "Principal component models for sparse functional data ",
	journal = "Biometrika",
	volume = "87",
	pages = "587--602",
	year = "2000",
	doi = "10.1093/biomet/87.3.587",
	author = "Gareth M. James and Trevor J. Hastie and Catherine A. Sugar ",
}

@article{KramerSugiyama2011,
    author = {Nicole Kr\"{a}mer and Masashi Sugiyama},
    title = {The Degrees of Freedom of Partial Least Squares Regression},
    journal = jasa,
    volume = {106},
    pages = {697--705},
    year  = {2011},
    doi = {10.1198/jasa.2011.tm10107},
}

@book{Lange2010,
   Author = {Kenneth Lange},
   Title = {Numerical Analysis for Statisticians},
   Publisher = {Springer},
   Address = {New York},
   edition = {2nd},
   Year = {2010},
   doi = {10.1007/978-1-4419-5945-4},
}

@article{LiHsing2010,
    author = "Li, Yehua and Hsing, Tailen",
    doi = "10.1214/10-AOS813",
    journal = aos,
    pages = "3321--3351",
    title = "Uniform convergence rates for nonparametric regression and principal component analysis in     functional/longitudinal data",
    volume = "38",
    year = "2010"
}

@article{Marx1996,
	author = {Brian D. Marx},
	journal = {Technometrics},
	pages = {374--381},
	title = {Iteratively Reweighted Partial Least Squares Estimation for Generalized Linear Regression},
	volume = {38},
	year = {1996},
}

@article{PaulPeng2011,
    author = "Paul, Debashis and Peng, Jie",
    doi = "10.1214/11-EJS662",
    journal = ejs,
    pages = "1960--2003",
    title = "Principal components analysis for sparsely observed correlated functional data 
        using a kernel smoothing approach",
    volume = "5",
    year = "2011",
}

@article{PengPaul2009,
    author = {Jie Peng and Debashis Paul},
    title = {A Geometric Approach to Maximum Likelihood Estimation of 
        the Functional Principal Components From Sparse Longitudinal Data},
    journal = jcgs,
    volume = {18},
    pages = {995--1015},
    year  = {2009},
    doi = {10.1198/jcgs.2009.08011},
}

@article{PredaSaporta2005,
    author = {C. Preda and G. Saporta},
    title = {{PLS} regression on a stochastic process},
    journal = csda,
    Volume = {48},
    year = {2005},
    pages = {149--158},
    doi = {10.1016/j.csda.2003.10.003},
}

@article{ReissOgden2007,
    author = {Philip T Reiss and R. Todd Ogden},
    title = {Functional Principal Component Regression and Functional Partial Least Squares},
    journal = jasa,
    volume = {102},
    pages = {984--996},
    year  = {2007},
    doi = {10.1198/016214507000000527},
}

@article{RubinPanaretos2020,
    author = "Rub\'{i}n, Tom\'{a}\v{s} and Panaretos, Victor M.",
    doi = "10.1214/20-EJS1690",
    journal = ejs,
    pages = "1137--1210",
    title = "Sparsely observed functional time series: estimation and prediction",
    volume = "14",
    year = "2020",
}

@article{Tasaki2009,
    title = "Convergence rates of approximate sums of Riemann integrals",
    journal = jat,
    volume = "161",
    pages = "477--490",
    year = "2009",
    doi = "10.1016/j.jat.2008.10.005",
    author = "Hiroyuki Tasaki",
}

@Book{TherneauGrambsch2000,
    title = {Modeling Survival Data: Extending the {C}ox Model},
    author = {Terry M. Therneau and Patricia M. Grambsch},
    year = {2000},
    publisher = {Springer},
    address = {New York},
    doi = {10.1007/978-1-4757-3294-8},
}

@article{Tombaugh2006,
    title = "A comprehensive review of the {Paced Auditory Serial Addition Test (PASAT)}",
    journal = "Arch. Clin. Neuropsych.",
    volume = "21",
    pages = "53--76",
    year = "2006",
    doi = "10.1016/j.jat.2008.10.005",
    author = "Tom N Tombaugh",
}

@article{WangIbrahimZhu2020,
	author = {Wang, Yue and Ibrahim, Joseph G. and Zhu, Hongtu},
	title = {Partial least squares for functional joint models with applications to the Alzheimer's disease neuroimaging initiative study},
	journal = {Biometrics},
	volume = {},
	number = {},
	pages = {},
	year = {2020},
	doi = {10.1111/biom.13219},
	note = {in press},
}

@incollection{Wold1975,
    author = {Herman Wold},
    title = {Path models with latent variables: the {NIPALS} approach},
    editor = {H.M. Blalock and A. Aganbegian and F. M. Borodkin and Raymond Boudon and Vittorio Capecchi},
    booktitle = {Quantitative Sociology: 
        International Perspectives on Mathematical and Statistical Model Building},
    year = {1975},
    publisher = {Academic Press},
    address = {New York},
    pages = {307--335},
}

@article{YaoMullerWang2005a,
    author = "Yao, Fang and M{\"{u}}ller, Hans-Georg and Wang, Jane-Ling",
    doi = "10.1198/01621450400001745",
    journal = jasa,
    pages = "577--590",
    title = "Functional data analysis for sparse longitudinal data",
    volume = "100",
    year = "2005",
}

@article{YaoMullerWang2005b,
    author = "Yao, Fang and M{\"{u}}ller, Hans-Georg and Wang, Jane-Ling",
    doi = "10.1214/009053605000000660",
    journal = aos,
    pages = "2873--2903",
    title = "Functional linear regression analysis for longitudinal data",
    volume = "33",
    year = "2005",
}

@article{XiaoLiCheckleyCrainiceanu2018,
    author = "Luo Xiao and Cai Li and William Checkley and Ciprian Crainiceanu",
    doi = "10.1007/s11222-017-9744-8",
    journal = "Stat. Comput.",
    pages = "511--522",
    title = "Fast covariance estimation for sparse functional data",
    volume = "28",
    year = "2018",
}

@article{Zhou2019,
    title = "Functional continuum regression",
    journal = jmva,
    volume = "173",
    pages = "328--346",
    year = "2019",
    doi = "10.1016/j.jmva.2019.03.006",
    author = "Zhiyang Zhou",
}

@unpublished{Zhou2020,
	author = {Zhiyang Zhou},
	title = {Partial least squares for function-on-function regression via {K}rylov subspaces},
	year = {2020},
	note = "arXiv:2005.04798",
}

@article{ZhouLinLiang2018,
    author = {Ling Zhou and Huazhen Lin and Hua Liang},
    title = {Efficient Estimation of the Nonparametric Mean and Covariance Functions for Longitudinal and     Sparse Functional Data},
    journal = jasa,
    volume = {113},
    pages = {1550--1564},
    year  = {2018},
    doi = {10.1080/01621459.2017.1356317},
}

@Manual{R,
    title = {R: A Language and Environment for Statistical Computing},
    author = {{R Core Team}},
    organization = {R Foundation for Statistical Computing},
    address = {Vienna, Austria},
    year = {2020},
    url = {https://www.R-project.org/},
	note = {R version 3.4.2 ``Short Summer''},
}

@Manual{R-face,
    title = {face: Fast Covariance Estimation for Sparse Functional Data},
    author = {Luo Xiao and Cai Li and William Checkley and Ciprian Crainiceanu},
    year = {2019},
    note = {R package version 0.1-5},
    url = {https://CRAN.R-project.org/package=face},
}

@Manual{R-fdapace,
    title = {fdapace: Functional Data Analysis and Empirical Dynamics},
    author = {Cody Carroll and Alvaro Gajardo and Yaqing Chen and Xiongtao Dai and Jianing Fan and Pantelis Z. 
        Hadjipantelis and Kyunghee Han and Hao Ji and Hans-Georg Mueller and Jane-Ling Wang},
    year = {2020},
    note = {R package version 0.5.4},
    url = {https://CRAN.R-project.org/package=fdapace},
}

@Manual{R-orthopolynom,
    title = {orthopolynom: Collection of functions for orthogonal and orthonormal polynomials},
    author = {Frederick Novomestky},
    year = {2013},
    note = {R package version 1.0-5},
    url = {https://CRAN.R-project.org/package=orthopolynom},
}

@Manual{R-refund,
    title = {refund: Regression with Functional Data},
    author = {Jeff Goldsmith and Fabian Scheipl and Lei Huang and Julia Wrobel and Chongzhi Di and 
        Jonathan Gellar and Jaroslaw Harezlak and Mathew W. McLean and Bruce Swihart and Luo Xiao and 
        Ciprian Crainiceanu and Philip T. Reiss},
    year = {2019},
    note = {R package version 0.1-21},
    url = {https://CRAN.R-project.org/package=refund},
}

@Manual{R-survival,
    title = {A Package for Survival Analysis in R},
    author = {Terry M Therneau},
    year = {2020},
    note = {R package version 3.2-3},
    url = {https://CRAN.R-project.org/package=survival},
}
\end{document}